\documentclass[12pt]{article}
\usepackage[big]{dg_arxiv}
\usepackage[utf8]{inputenc}
\usepackage{epsfig,bm,epsf,float}
\usepackage{amsmath,amsthm,amsfonts,amssymb}
\usepackage{graphicx, caption, subcaption}
\usepackage{float}
\usepackage{bbm}
\usepackage{color}

\makeatletter
\let\my@saved@original@eqref\eqref
\renewcommand*{\eqref}[1]{\ignorespaces
	\begingroup
	\let\normalfont\relax
	\my@saved@original@eqref{#1}\ignorespaces
	\endgroup
}
\makeatother

\begin{document}

\def\cal#1{\mathcal{#1}}
\newcommand{\mbb}{\mathbb}
\def\mbb#1{\mathbb{#1}}

\def\be{\begin{equation*}}
\def\ee{\end{equation*}}

\newcommand{\wk}{(\Omega,\cal A,\mbb P)}
\newcommand{\J}{J}
\newcommand{\mbf}[1]{{\mathbf #1}}
\newcommand{\bs}[1]{{\boldsymbol #1}}
\def\1#1{\mathbbm 1_{\{#1\}}}
\theoremstyle{plain}
\newtheorem{theorem}{Theorem}[section]
\newtheorem{cor}[theorem]{Corollary}
\newtheorem{lemma}[theorem]{Lemma}
\newtheorem{proposition}[theorem]{Proposition}
\newtheorem{axiom}{Axiom}
\newtheorem{openprob}{Open problem}

\theoremstyle{definition}
\newtheorem{assumption}[theorem]{Assumption}
\newtheorem{definition}[theorem]{Definition}
\newtheorem{example}[theorem]{Example}
\newtheorem{remark}[theorem]{Remark}
\newtheorem{algorithm}[theorem]{Algorithm}

\renewcommand{\d}{{\text{d}}}

\newcommand{\VaR}{\text{VaR}}
\newcommand{\TCE}{\text{TCE}}
\newcommand{\EL}{\text{EL}}
\newcommand{\Loss}{L}
\newcommand{\myparagraph}[1]{\textbf{#1}}

\graphicspath{./}	
	
\articletype{~}

\author[1]{Imke Redeker}
\author[2]{Ralf Wunderlich}
\runningauthor{Redeker and Wunderlich}
\affil[1]{
German Rheumatism Research Centre, Epidemiology Unit,\newline Charitéplatz 1, 10117 Berlin, Germany, 
\texttt{imke.redeker@drfz.de }}
\affil[2]{Institute of Mathematics, Brandenburg University of Technology \newline Cottbus - Senftenberg, Postbox 101344, 03013 Cottbus, Germany, \texttt{ralf.wunderlich@b-tu.de}\newline
The authors thank Michaela Szölgyenyi (Klagenfurt) and Rüdiger Frey (Vienna) for valuable discussions on the topic of the paper. }

\title{Credit risk with asymmetric information and a switching default threshold}
\runningtitle{ Credit risk with asymmetric information and a switching default threshold}
\subtitle{~}
\abstract{
We investigate the impact of available information on the estimation of the default probability within a generalized structural model for credit risk. The traditional structural model where default is triggered when the value of the firm's asset falls below a constant threshold is extended by relaxing the assumption of a constant default threshold. The default threshold at which the firm is liquidated is modeled as a random variable whose value is chosen by the management of the firm and dynamically adjusted to account for changes in the economy or the appointment of a new firm management.
Investors on the market have no access to the value of the threshold and only anticipate the distribution of the threshold. We distinguish different information levels on the firm's assets and derive explicit formulas for the conditional default probability given these information levels. Numerical results indicate that the information level has a considerable impact on the estimation of the default probability and the associated credit yield spread.
}
\classification[MSC]{Primary 91G40;  Secondary 91B70, 91G50}
\keywords{Credit risk, Structural model, Asymmetric information, Switching default threshold}

\journalname{}
  \journalyear{~}
\journalvolume{}
\journalissue{}
\startpage{1}

\maketitle

\section{Introduction}
Credit risk, or default risk, is the risk that a financial loss will be incurred if a counterparty does not fulfill its contractually agreed financial obligations in a timely manner. Quantitative credit risk models for measuring, monitoring and managing credit risk have become central in today's complex financial industry. The recent financial crisis has impressively demonstrated the need for effective credit risk management. Since then the evaluation of credit risk has been receiving increasing attention. For credit risk analysis it is crucial both from a theoretical and an empirical point of view to model the default of a default risky asset, i.e., a security that has a nonzero probability of defaulting on its contracted payments, and to forecast the associated default probability. A typical example of a credit risky asset is a corporate bond. A corporate bond promises its holder a fixed stream of payments but may default on its promise.\\ 
There are two classical types of modeling approaches for credit risk: the structural one and the reduced-form one. The structural approach is considered by \textsc{Black \& Scholes} \cite{black}, \textsc{Merton} \cite{merton_cr} and \textsc{Black \& Cox} \cite{BlackCox}, among others. It provides a relationship between default risk and capital structure by using the evolution of the firm’s assets value to determine the time of default, i.e., the default event of a bond is triggered when the assets of the firm who issued the bond fall below some threshold. The important feature of the structural model is that it implicitly assumes that the modeler has complete knowledge about the dynamics of the firm's assets and the situation that will trigger the default event (i.e., the firm's liabilities). Despite the convincing economic interpretation in terms of the firm's assets and liabilities there are shortcomings when the firm's assets are modeled by a continuous-time asset value process. One is that credit yield spreads go to zero as maturity goes to zero regardless of the riskiness of the firm. This results from the investors' knowledge about the firm's true distance to default. Such credit spreads are uncommon in practice. Another disadvantage of the structural approach is that forecast bond prices continuously converge to their recovery value (the payment which is received if default occurs before maturity) which contradicts the price jump at default in empirical studies.
These issues do not occur in the second approach, the reduced-form approach, which is considered by \textsc{Jarrow \& Turnbull} \cite{jarrow_turnbull}, \textsc{Artzner \& Delbaen} \cite{ArtznerDelbaen}, and \textsc{Duffie \& Singleton} \cite{DuffieSingleton}, among others. It treats the dynamics of default as an exogenous event. This implies knowledge of a less detailed information set compared to the structural approach and credit spreads become in general more realistic and are easier to quantify. Another advantage is that the reduced-form approach has proven to be very useful for the valuation of credit-sensitive securities. However, the approach is lacking economic insights as it does not connect credit risk to underlying structural variables.
To gain both the economic appeal of the structural approach and the empirical plausibility and the tractability of the reduced-form approach, structural models can be transformed into reduced-form models by changing its information set to a less refined one (see \textsc{Jarrow \& Protter} \cite{jarrowprotter}). One way is to model the default barrier as a random variable which is unobservable by bond investors (see \textsc{Lando} \cite{lando}, \textsc{Giesecke \& Goldberg} \cite{giesecke-goldberg}, \textsc{Hillairet \& Jiao} \cite{hillairet-jiao}). Another way is to assume that the firm's assets are only partially observable by investors (see \textsc{Duffie \& Lando} \cite{duffie-lando}, \textsc{Jeanblanc \& Valchev} \cite{jeanblanc-valchev}, \textsc{Lakner \& Liang} \cite{lakner-liang}).\\
This paper extends the traditional structural model to a dynamic setting by relaxing the assumption of a constant default threshold to the case of a piecewise constant threshold. Motivated by default events during the financial crisis where a firm's management has decided to close activities and the default occurs although the firm is in a relatively healthy situation (see \textsc{Hillairet \& Jiao} \cite{hillairet-jiao}), the default threshold is modeled as a random variable whose value is chosen by the management of the firm and adjusted dynamically to react to changes in the economic environment or to account for the election of a new firm management. In literature, this generalization of the default model to a dynamic setting was proposed by \textsc{Blanchet-Scalliet, Hillairet \& Jiao} \cite{blanchet-hillairet}. The authors study the information accessible to the management of the firm and obtain explicit formulations for the survival probability given the information of the management by using a successive enlargement framework. Our approach is different and related to ordinary investors on the market who do not have access to the value of the threshold and only anticipate the distribution of the threshold. The objective is to analyze the impact of available information of public bond investors on the estimation of the default probability. We consider an investor who continuously observes the firm value and an investor who only observes the firm value at discrete dates. Explicit formulas for the default probabilities and associated credit yield spreads given the different information levels on the firm's assets are derived and, based on these formulas, as a direct application the valuation of bond prices is considered. Numerical examples are presented to illustrate and compare the conditional default probabilities and credit spreads.\\
The remainder of this paper is organized as follows. Section~\ref{sec:model} sets up a model for credit risk based on a structural model but with an unobservable default barrier that is allowed to switch. In this setting different information structures are distinguished. The impact of asymmetric information on the default probability is studied in Section~\ref{sec:con_prob} where explicit formulas for the conditional survival probabilities given the different information structures are derived. Section~\ref{sec:numeric} provides numerical results and a conclusion is given in Section~\ref{sec:conclusio}. Finally, the Appendix contains proofs omitted from the main text.

\section{Model for the default event}\label{sec:model}
We introduce a structural model for credit risk where short-term default risk is included by making the default barrier unobservable. Further, the usual assumption of a constant default barrier is relaxed allowing the firm's management to adjust the barrier. 
\subsection{Default barrier}\label{subsec:barrier}
\begin{figure}[b!]
	\centering
	\captionsetup{format=hang, justification=centerfirst, textfont=normalsize, labelfont=normalsize}
	\begin{subfigure}[b]{0.48\textwidth}
		\centering
		\includegraphics[width=\textwidth]{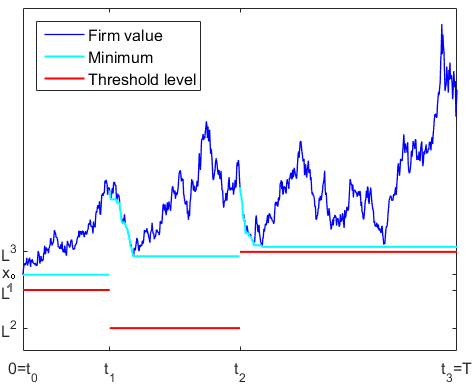}
		\caption{No default}
	\end{subfigure}
	\quad
	\begin{subfigure}[b]{0.48\textwidth}  
		\centering 
		\includegraphics[width=\textwidth]{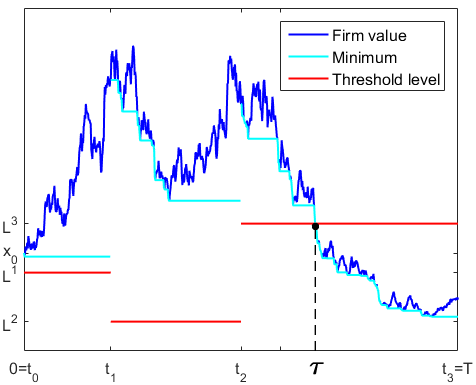}
		\caption{Default}
	\end{subfigure}
	\captionsetup{format=hang, justification=raggedright, textfont=normalsize, labelfont=normalsize}
	\caption[Firm value, running minimum process and switching barrier.]{Plot of two trajectories of the firm's asset process $X_t$, the associated running minimum process $M_{[t_{i-1},t)}$ and the default thresholds $L^i$.}\label{fig:model_setup_1}
\end{figure}
The default event is specified in terms of the firm's asset value process $(X_t)_{t\geq 0}$ and the default threshold $L$. Default occurs when the value of the firm decreases to the level of the default barrier for the first time. The random time of default is denoted by $\tau$. Uncertainty in the economy is modeled by some fixed complete probability space $\wk$ which is endowed with the filtration $\mbb F =(\cal F_t)_{t\geq 0}$ generated by the asset process $X$ and satisfying the usual conditions, i.e., $\cal F_t = \sigma(X_s \colon s\leq t) \vee \cal N$, where $\cal N$ denotes the $\mbb P$-null sets. The $\sigma$-algebra $\cal A$ of $\Omega$ describes the information of the market.
We model the evolution of the value $X$ of the firm's assets by a geometric Brownian motion, i.e.,
\begin{align}
\d X_t &= X_t(\mu \d t + \sigma \d B_t), \qquad X_0=x_0,\label{eq:firm_value}
\end{align}
where $(B_t)_{t \ge 0}$ is a $\mbb F$-Brownian motion, $\mu \in \mbb R$ and $\sigma >0$. The solution to \eqref{eq:firm_value} is known to be
\begin{align}
X_t &= x_0 e^{m t + \sigma B_t},\label{eq:firm_value_2}
\end{align}
where $m = \mu - \sigma^2/2$. For the sake of easier notation we assume w.l.o.g. $x_0=1$, i.e., we assume that the asset process starts at 1. We say that a firm defaults when it stops fulfilling a contractual commitment to meet its obligations stated in a financial contract. The firm's management decides whether and when to default. Thus, the management determines the default triggering barrier. The essential difference with a classical structural model is that the management is not constrained to decide on one fixed barrier but can dynamically adjust the default barrier. This reflects the management's possibility to react to changes in the economic environment or the election of a new firm management. The time points at which the management adjusts the default barrier are deterministic and denoted by $t_i$, $i=0,\ldots,n-1$, with $0 = t_0 < t_1<\ldots <t_{n-1}< T$, where $T=:t_n$ is a finite time horizon. The default barrier $L$ can be written as
\be
L_t = \sum_{i=1}^{n} L^i\mathbbm 1_{[t_{i-1},t_{i})},
\ee
where $L^i$, $i = 1,\ldots,n$ are $\cal A$-measurable random variables representing the private information of the management on the default barrier and $\mathbbm 1_A$ denotes the indicator function of the set $A$. Public investors do not have any knowledge on the default barrier except that they know the adjustment time points $t_0,\ldots,t_{n-1}$ and they agree on the joint probability distributions for $L^1,\ldots,L^i$, $i = 1,\ldots,n$, which are denoted by $F_{L^1,\ldots,L^i}$. The associated probability density functions are denoted by $f_{L^1,\ldots,L^i}$ for $i=1,\ldots,n$. We make the assumption that $\mathbf{L}=(L^1,\ldots,L^n)$ is independent of $\cal F_T$. 
The random default time is given by 
$\tau = \inf\{t > 0\colon X_t \leq L_t\}$.
The running minimum asset process is denoted by $M$ and given by
\begin{align}\label{eq:run_min}
M_t = \inf\{X_u \colon 0\leq u < t\}.
\end{align}
Further, the running minimum started from a certain time point $s$ is denoted by
\begin{align}\label{eq:run_min2}
M_{[s,t)} = \inf\{X_u \colon s\leq u<t\}.
\end{align}
For $s=0$ we have $M_{[0,t)}=M_t$. This model setup is illustrated by Figure \ref{fig:model_setup_1} which shows two trajectories of the firm value $X$, $n=3$ default barriers $L^1$, $L^2$ and $L^3$ and the three running minimum processes $M_{[t_{i-1},t_i)}$, for $i=1,2,3$. We consider two scenarios. In the first scenario shown in the left panel occurs no default whereas in the second scenario shown in the right panel default occurs between $t_2$ and $t_3=T$. 

\subsection{Information structure}\label{subsec:info}
We distinguish the following three information struc\-tures.\\
\textbf{{Management's information}}\\
The management has complete information about the firm's asset process $X$ and obtains information on the default threshold $L^i$ at time $t_{i-1}$. Thus, the management's information structure $\mbb G^M = (\cal G^M_t)_{t \in [0,T]}$ is a progressive enlargement of the filtration $\mbb F$ by the default threshold process $(L_t)_{t\geq 0}$, i.e.,
\be
\cal G_t^M = \cal F_t\vee \sigma(L_s,\, s\leq t).
\ee 
This insider information is considered in \textsc{Blanchet-Scalliet, Hillairet \& Jiao} \cite{blanchet-hillairet} and will not be considered in this paper.\\
For public bond investors we distinguish the following two information structures.\\
\textbf{{C-investor's information}}\\
The first type of investors \underline{c}ontinuously observe the firm value and the default in the moment it occurs but they do not have knowledge on the default threshold $L$ (because it is firm inside information of the management).  We call an investor endowed with this information structure a \textit{C-investor}. The C-investor's information structure $\mbb G^C = (\cal G^C_t)_{t \in [0,T]}$ on the bond market is described by a progressive enlargement of the filtration $\mbb F$ by the random default time $\tau$, i.e.,
\be 
\cal G^C_t = \cal F_t\vee \sigma(H_s,\, s\leq t),
\ee
where $H$ is the default indicator process defined by $H_t = \1{t \geq \tau}$. C-investors are uncertain about the firm's true distance to default although they have complete information about the firm value. This uncertainty is due to lacking knowledge on the threshold level. Thus, default arrives as a full surprise.\\
\textbf{{D-investor's information}}\\
The second type of investors observe the asset process only in \underline{d}iscrete  time and are called \textit{D-investors}.  The information structure of a D-investor is similar to the information structure of a C-investor in that both investors do not have any knowledge about the default barrier except that they observe the occurrence and timing of default. The difference is that the asset process $X$ is not completely observable but only at discrete dates denoted by $T_j$, $j=0,\ldots,J-1$, where $0=T_0< T_1<\ldots<T_{\J-1}< T_J=T$ for $\J \in \mbb N$. This is a realistic assumption, since investors usually observe the asset value at the times of corporate news release. The partial information on the asset process is described by a sub-filtration $\mbb F^D = (\cal F^D_t)_{t \in [0,T]}$ of $\mbb F$, where
\be
\cal F^D_t=\begin{cases}
	\cal F_0,  & \text{if } t<T_1,\\
	\sigma(X_{T_1},\ldots,X_{T_i}) & \text{if } T_i\leq t <T_{i+1},\, i\in\{1,\ldots,\J-1\}.
\end{cases}
\ee 
The D-investor's information structure $\mbb G^D = (\cal G^D_t)_{t \in [0,T]}$ can be described by a progressive enlargement of the filtration $\mbb F^D$ by the random default time $\tau$, i.e.,
\be 
\cal G^D_t = \cal F^D_t\vee \sigma(H_s,\, s\leq t).
\ee 
\begin{assumption}\label{assumption_d-investor}
	Every adjustment time $t_i$, $i=0,\ldots,n-1$, of the default barrier coincides with one of the information dates $T_j$, $j=0,\ldots,\J-1$, i.e., $\{t_0,\ldots,t_{n-1}\} \subseteq \{T_0,\ldots,T_{\J-1}\}$.
\end{assumption}

\section{Conditional survival probability}\label{sec:con_prob}
The conditional survival probability, i.e., the probability of not having experienced default by the finite time horizon $T$ given the accessible information, plays an important role in the valuation of credit risky securities (see \textsc{Hillairet \& Jiao} \cite{hillairet-jiao2}). The aim of this section is to derive explicit formulas for the conditional survival probability given the information of a C-investor, i.e., $\mbb P(\tau > T|\cal G^C_t)$, and of a D-investor, i.e., $\mbb P(\tau > T|\cal G^D_t)$.\\
We begin with reviewing classical results of the running minimum of a geometric Brownian motion.
\begin{lemma}\label{distr_runmin}	
	Let $X=(X_t)_{t\ge 0}$ be the asset value process given in \eqref{eq:firm_value} with $X_0=1$ and $M=(M_t)_{t\geq 0}$ the running minimum process defined in \eqref{eq:run_min}. 
	\begin{enumerate}
	\item Given $t >0$ then the density function $f^M_{t}(\cdot)$ of $M_t$ is given by
	\begin{align}\label{eq:density_min}
	f^M_{t}(w)=&\frac{1}{\sigma w \sqrt{t}}\,\varphi \Big(\frac{mt-\frac{\ln{w}}{\sigma}}{\sqrt{t}}\Big)+\frac{e^{2m\frac{\ln{w}}{\sigma}}}{\sigma w \sqrt{t}}\, \varphi\Big(\frac{mt+\frac{\ln{w}}{\sigma}}{\sqrt{t}}\Big)+\frac{2m}{w \sigma}e^{2m\frac{\ln{w}}{\sigma}}\,\varPhi \Big(\frac{mt+\frac{\ln{w}}{\sigma}}{\sqrt{t}}\Big)
	\end{align}
	for $w\in (0,1]$ and zero otherwise, where $\varPhi$ and $\varphi$ denote the cumulative distribution function and the probability density function of the standard normal distribution, respectively. 
	\item Given $t >0$ then the joint probability density function $f^{M,X}_{t}(\cdot, \cdot)$ of $X_t$ and $M_t$ is given by
	\begin{align}\label{eq:joint_density_min_x}
	f^{M,X}_{t}(u,v)= \frac{2v^{m/\sigma^2 -1}\ln(v/u^2)}{\sigma^3 \sqrt{2\pi}t^{3/2}u} e^{-\frac{m^2t}{2\sigma^2}}e^{-\frac{\ln^2(v/u^2)}{2\sigma^2 t}}
	\end{align}
	for $u\in (0,1]$, $u\leq v$, and zero otherwise.	
	\end{enumerate}
\end{lemma}	
\begin{proof}
	The above formulas are a corollary to the results given in \textsc{Harrison} \cite[Ch. 1]{harrison}.
\end{proof}
\begin{lemma}\label{compl_distr}
	Under the assumptions of Lemma \ref{distr_runmin} the complementary distribution function $\Psi$ of $M$ is given by  	
	\begin{align}\label{eq:distribution_min_x}
	\Psi(t,u)&=\mbb P(M_t > u)=\varPhi\left(\frac{-\ln(u)+mt}{\sigma\sqrt{t}} \right)-\exp\left\{\frac{2m}{\sigma^2}\ln(u) \right\}\varPhi\left(\frac{\ln(u)+mt}{\sigma\sqrt{t}} \right)
	\end{align}
	for $t>0$, $u\leq 1$. Further, it holds $\Psi(t,u)=0$ for $t \geq 0$, $u> 1$ and $\Psi(0,u)=1$ for $u\leq 1$.
\end{lemma}	  
\begin{proof}
	The proof is given in \textsc{Jeanblanc \& Valchev} \cite[Lemma 2]{jeanblanc-valchev}.
\end{proof}
\subsection{C-investor's information case}\label{sec_cr:s_inv}
The next theorem shows that the conditional survival probability given the information of a C-investor can be formulated in terms of $\mbb F$-conditional survival probabilities for which we derive explicit formulas. 
\begin{theorem}\label{theo_C_investor}
	Let $X=(X_t)_{t\ge 0}$ be the asset value process given in \eqref{eq:firm_value} with $X_0=1$. Further, let $M_t$ be the associated running minimum started at zero defined in \eqref{eq:run_min} and $M_{[s,t)}$ be the associated running minimum started at a certain time point $s$ defined in \eqref{eq:run_min2}. Then, for the conditional survival probability given the information of a C-investor it holds
	\begin{align}\label{eq:enl_fil}
	\mbb P(\tau > T|\cal G^C_t) = \1{\tau >t }\frac{\mbb P(\tau > T|\cal F_t)}{\mbb P(\tau > t|\cal F_t)}, \qquad \text{for } t < T,
	\end{align}
	where the $\mbb F$-conditional survival probabilities are given by the following formulas for $n\geq 2$:
	\begin{enumerate}
		\item For $t\in [t_{i-1},t_i)$, $i=1,\ldots,n-1$, we have
		\begin{align}
		\begin{split}\label{eq:sp_s_a}
		\mbb P(\tau > T|\cal F_t)&= \iint\limits_{[0,1]\times [u,\infty)}\ldots \iint\limits_{[0,1]\times [u,\infty)}  f^{\widehat M, \widehat X}_{t_i-t}(u_i,v_i)\prod\limits_{j=i+1}^{n-1}f^{\widehat M, \widehat X}_{t_j-t_{j-1}}(u_j,v_j)\\
		&\quad \int_0^1F_{L^1,\ldots,L^n}( M_{t_1},M_{[t_1,t_2)},\ldots,M_{[t_{i-2},t_{i-1})},\min(M_{[t_{i-1},t)},u_iX_t),\\
		&\qquad\quad  u_{i+1}v_iX_t,\ldots,u_{n-1}v_{n-2}\ldots v_iX_t, wv_{n-1}\ldots v_iX_t)f^{\widehat M}_{T-t_{n-1}}(w)\\
		&\qquad \quad \d w\d v_{n-1} \d u_{n-1}\ldots \d v_i \d u_i,
		\end{split}
		\end{align}
		where $f_t^{\widehat M}$ and $f_t^{\widehat M, \widehat X}$ are given in \eqref{eq:density_min} and \eqref{eq:joint_density_min_x}, respectively.\\
		For $t\in [t_{n-1},T)$ it holds
		\begin{align}\label{eq:sp_s_b}
		\mbb P(\tau > T|\cal F_t)&= \int_0^1 F_{L^1,\ldots,L^n}(M_{t_1},M_{[t_1,t_2)},\ldots,M_{[t_{n-2},t_{n-1})},\min(M_{[t_{n-1},t)},wX_t)) f^{\widehat M}_{T-t}(w)\d w.
		\end{align} 
		\item For $t\in [t_{i-1},t_i)$, $i=1,\ldots,n$, it holds
		\begin{align}\label{eq:sp_s_c}
		\mbb P(\tau >t|\cal F_t)=F_{L^1,\ldots,L^i}(M_{t_1},M_{[t_1,t_2)},\ldots,M_{[t_{i-2},t_{i-1})},M_{[t_{i-1},t)}).
		\end{align}		
	\end{enumerate}
\end{theorem}
\begin{proof}
	We obtain Eq. \eqref{eq:enl_fil} by using classical results of progressive enlargement (see \textsc{Jeanblanc, Yor \& Chesney} \cite[Sec. 7.3.3]{jeanblanc-yor}).
	For the sake of simpler notation the proof of the $\mbb F$-conditional survival probabilities is only given for $n=2$, i.e., the threshold is $L^1$ in the interval $[t_0,t_1)$ and $L^2$ in the interval $[t_1,T)$. The proof for $n>2$ is along the same line and skipped.\\
	Eq. \eqref{eq:firm_value_2} yields that for $s>t$ the firm value $X_s$ can be expressed by
	\begin{align*} 
	X_s &= X_t \exp\{m(s-t)+\sigma (B_s -B_t) \}= X_t \exp\{m(s-t)+\sigma (B_{s-t+t} -B_t) \}\\
	&= X_t \exp\{m(s-t)+\sigma \widehat B_{s-t} \} = X_t Y_{s-t},
	\end{align*}	
	where $(\widehat B_u)_{u\geq 0}$ given by $\widehat B_u = B_{u+t}-B_t$ is a Brownian motion starting at zero and independent of $\cal F_t$. The process $(Y_u)_{u\geq 0}$ given by $Y_u = \exp\{mu+\sigma \widehat B_{u} \}$ is independent of $\cal F_t$ and it holds
	\begin{align*}
	Y_u = \frac{X_{t+u}}{X_t} \overset{d}{=}X_u.
	\end{align*}
\myparagraph{Proof of \eqref{eq:sp_s_a}:}  Let $t\in [t_0,t_1)$ be fixed, then we can describe the event that no default occurs until the maturity time $T$ by
		\begin{align*}
		\{\tau > T\} &= \{L^1 < M_{t_1}\} \cap  \{L^2 < M_{[t_1,T)} \} = \{L^1 < \inf_{s< t_1}X_s\} \cap  \{L^2 < \inf_{t_1\leq s <T} X_s\}\\
		&= \{L^1 < \inf_{s< t}X_s\} \cap \{L^1 < \inf_{t\leq s< t_1}X_s\} \cap  \{L^2 < \inf_{t_1\leq s <T} X_s\}\\
		&=\{L^1 < \inf_{s< t}X_s\} \cap \{L^1 < \inf_{t\leq s< t_1}X_tY_{s-t}\} \cap  \{L^2 < \inf_{t_1\leq s <T} X_{t_1}Z_{s-t_1}\}, 
		\end{align*}  
		where $(Z_u)_{u\geq 0}$ given by $Z_u = \exp\{mu+\sigma \widetilde B_{u} \}$ with the Brownian motion $(\widetilde B_u)_{u\geq 0}$ given by $\widetilde B_u = B_{u+t_1}-B_{t_1}$ is independent of $\cal F_{t_1}$ and it holds $Z_u \overset{d}{=}X_u$. We denote by $(\widehat M_u)_{u\geq 0}$ and $(\widetilde M_u)_{u\geq 0}$ the running minimum of $Y$ and $Z$, respectively, i.e.,
		\be 
		\widehat M_u = \inf_{s < u} Y_s \qquad \text{and}\qquad \widetilde M_u = \inf_{s < u} Z_s.
		\ee 
		Then it holds
		\begin{align*}
		\{\tau > T\} &= \{L^1 < M_{t}\} \cap \{ L^1 < \widehat M_{t_1-t}X_t\} \cap \{ L^2 < \widetilde M_{T-t_1}X_{t_1} \}\\
		&=\{L^1 < M_{t}\} \cap \{ L^1 < \widehat M_{t_1-t}X_t\} \cap \{ L^2 < \widetilde M_{T-t_1}Y_{t_1-t}X_t \}.
		\end{align*}
		Based on this representation the conditional survival probability until the maturity time $T$ given the information $\cal F_t$ can be written as
		\begin{align*}
		\mbb P(\tau >T|\cal F_t)&=\mbb P(L^1 < M_{t}, L^1 < \widehat M_{t_1-t}X_t, L^2 < \widetilde M_{T-t_1}Y_{t_1-t}X_t|\cal F_t)\\
		&=\int_0^1 \int_{u}^{\infty} \mbb P(L^1 < M_{t}, L^1 < uX_t, L^2 < \widetilde M_{T-t_1}vX_t|\cal F_t) f^{\widehat M,Y}_{t_1-t}(u,v)\d v \d u\\
		&=\int_0^1\int_u^{\infty} \int_{0}^{1} \mbb P(L^1 < \min{(M_{t},uX_t)}, L^2 < wvX_t|\cal F_t)f^{\widetilde M}_{T-t_1}(w) f^{\widehat M,Y}_{t_1-t}(u,v)\d w \d v \d u\\
		&=\int_0^1\int_u^{\infty} \int_{0}^{1}F_{L^1,L^2}(\min{(M_{t},uX_t)}, wvX_t)f^{\widetilde M}_{T-t_1}(w)f^{\widehat M,Y}_{t_1-t}(u,v)\d w \d v \d u.
		\end{align*}		
		We have exploited in the last equation the independence of $(L^1,L^2)$ from $\cal F_T$.
		\\
		\myparagraph{Proof of \eqref{eq:sp_s_b}:} 
		Let $t\in [t_1,T)$ be fixed, then we can describe the event that no default occurs until the maturity time $T$ by
		\begin{align*}
		\{\tau > T\} &=\{L^1 < M_{t_1}\} \cap  \{L^2 < M_{[t_1,T)} \}= \{L^1 < M_{t_1}\} \cap \{L^2 < \inf_{t_1\leq s< t}X_s\} \cap  \{L^2 < \inf_{t\leq s <T} X_s\}\\
		&=\{L^1 < M_{t_1}\} \cap \{L^2 < \inf_{t_1\leq s< t}X_s\} \cap  \{L^2 < \inf_{t\leq s <T} X_{t}Y_{s-t}\}\\
		&=\{L^1 < M_{t_1}\} \cap \{ L^2 < M_{[t_1,t)}\} \cap \{ L^2 < \widehat M_{T-t}X_{t}\}.
		\end{align*}
		Based on this representation the conditional survival probability until the maturity time $T$ given the information $\cal F_t$ can be calculated by
		\begin{align*}
		\mbb P(\tau >T|\cal F_t)
		&=\mbb P(L^1 < M_{t_1}, L^2 < M_{[t_1,t)}, L^2 < \widehat M_{T-t}X_{t}|\cal F_t)\\
		&=\int_0^1 \mbb P(L^1 < M_{t_1}, L^2 < M_{[t_1,t)}, L^2 <w X_{t} |\cal F_t)f^{\widehat M}_{T-t}(w)\d w\\
		&=\int_0^1 F_{L^1,L^2}(M_{t_1}, \min(M_{[t_1,t)},w X_{t}))f^{\widehat M}_{T-t}(w)\d w.
		\end{align*}
		\myparagraph{Proof of \eqref{eq:sp_s_c}:} 
		The conditional survival probability until time $t\in [t_0,t_1)$ given the information $\cal F_t$ is obtained by		
		\begin{align*}
		\mbb P(\tau >t | \cal F_t) = \mbb P(L^1 < M_t | \cal F_t) = F_{L^1}(M_t).
		\end{align*}
		Finally, for $t\in [t_1,T)$ it holds
		\begin{align*}
		\mbb P(\tau >t | \cal F_t) = \mbb P(L^1 < M_{t_1}, L^2 < M_{[t_1,t)} | \cal F_t) = F_{L^1,L^2}(M_{t_1},M_{[t_1,t)} ).
		\end{align*}
\end{proof}
\begin{remark}
	For $n=1$ we obtain as a special case the model proposed by \textsc{Giesecke \& Goldberg} \cite{giesecke-goldberg}, where the default barrier is constant but random, i.e., $L_t=L^1$. The conditional survival probability is given by 
	\begin{align}\label{eq:cond_prob_s_giesecke}
	\mbb P(\tau > T|\cal G^C_t)
	= \1{\tau >t }\frac{\int_0^1  F_{L^1}(\min(M_t,w X_t))f^{\widehat M}_{T-t}(w)\d w}{F_{L^1}(M_t )}, \quad t\in [0,T).
	\end{align}
\end{remark}
\begin{remark}\label{rem:cond_prob_s}
	Let $n=2$ and assume that $L^1$ and $L^2$ are independent with probability distribution functions $F_{L^1}$ and $F_{L^2}$, respectively. Then the conditional survival probability for $t\in [t_0,t_1)$ is given by 
	\begin{align*}
	\mbb P(\tau >T|\cal G^C_t) &= \frac{1}{F_{L^1}(M_t)}\int_0^1\int_u^{\infty} \int_{0}^{1}F_{L^1}(\min{(M_{t},uX_t)})F_{L^2}(wvX_t)f^{\widetilde M}_{T-t_1}(w)f^{\widehat M,Y}_{t_1-t}(u,v)\d w \d v \d u
	\end{align*}
	on the no default set $\{\tau >t\}$. For $t\in [t_1,T)$ calculating the conditional survival probability reduces to the case of a constant but random barrier. We have
	\begin{align*}
	\mbb P(\tau >T|\cal F_t) 
	&=\int_0^1 F_{L^1,L^2}(M_{t_1}, \min(M_{[t_1,t)},w X_{t}))f^{\widehat M}_{T-t}(w)\d w\\
	&=F_{L^1}(M_{t_1})\int_0^1  F_{L^2}(\min(M_{[t_1,t)},w X_{t}))f^{\widehat M}_{T-t}(w)\d w
	\end{align*}
	and
	\begin{align*}
	\mbb P(\tau >t | \cal F_t)  = F_{L^1,L^2}(M_{t_1},M_{[t_1,t)} ) = F_{L^1}(M_{t_1})F_{L^2}(M_{[t_1,t)} )
	\end{align*}	
	yielding the following simplified formula for the conditional survival probability
	\begin{align*}
	\mbb P(\tau > T|\cal G^C_t) 
	 = \1{\tau >t }\frac{\int_0^1  F_{L^2}(\min(M_{[t_1,t)},w X_{t}))f^{\widehat M}_{T-t}(w)\d w}{F_{L^2}(M_{[t_1,t)} )}.
	\end{align*}	
\end{remark}
A direct application of the conditional survival probability is the pricing of credit derivatives such as defaultable bonds. For example, let us consider a zero-coupon bond that matures at $T$ and has zero recovery, i.e., the defaultable bond pays 1 at $T$ if there was no default by $T$ and zero otherwise. Assuming that the pricing probability is $\mbb P$, then the price $C_t$ of such a financial product is given by
\begin{align*}
C_t &= e^{-r(T-t)}\mbb E[\1{\tau >T}\,|\,\cal G^C_t]=e^{-r(T-t)}	\mbb P(\tau > T|\cal G^C_t) =e^{-r(T-t)} \1{\tau >t }\frac{\mbb P(\tau > T|\cal F_t)}{\mbb P(\tau > t|\cal F_t)},
\end{align*}
where $r\geq 0$ is a discount factor. An important quantity in the credit risk analysis is the credit yield spread $S_t$ on a zero-coupon bond issued by a firm. It is the difference between the yield at time $t$ on a credit risky and a credit risk-free zero-coupon bond, both maturing at $T$. Thus, the credit spread $S_t$ is given by
\begin{align*}
S_t=-\frac{1}{T-t}\ln\{\mbb P(\tau > T|\cal G^C_t)\}.
\end{align*}
\subsection{D-investor's information case}\label{sec_cr:d_inv}
n this subsection we consider the case where the asset process $X$ is not completely observable by ordinary investors on the market. More precisely, the D-investor obtains information about the asset value only at discrete times which include the adjustment times $t_k$, $k\in \{0,\ldots,n-1\}$, of the default barrier. We denote the information dates between two adjustment times $t_k$ and $t_{k+1}$ by $T_i^k$, $k=0,\ldots,n-1$, $i=0,\ldots,\J_k-1$, where $T_0^k :=t_k$, $T_{i-1}^k < T_i^k$, $T_{\J_k-1}^k <t_{k+1}$, $T_{\J_k}^k := t_{k+1}$ and $\sum_{k=0}^{n-1}\J_k =\J$. Further, we introduce 
\begin{align}\label{eq:cp_d_k}
K_j^k(\ell)=\mbb P(M_{[T^k_{j-1},T^k_j)}>\ell| X_{T^k_{j-1}},X_{T^k_j})\quad \text{and} \quad K^{k,i}(\ell)= \prod_{j=1}^i K_j^k(\ell)
\end{align}
for $k=0,\ldots,n-1$, $j=1,\ldots,\J_k$.
\begin{lemma}\label{lem:cond_prob_min}
	The conditional probabilities $K_j^k(\ell)$, $k=0,\ldots,n-1$, $j=1,\ldots,\J_k$, defined in \eqref{eq:cp_d_k} are given by
	\begin{align*}
	K^k_j(\ell) =  
	1-\exp\bigg\{\frac{-2}{\sigma^2 (T^k_j-T^k_{j-1})}\ln\bigg(\frac{\ell}{X_{T^k_{j-1}}}\bigg)\ln\bigg(\frac{\ell}{X_{T^k_j}}\bigg)\bigg\},
	\end{align*}
	for $\ell< \min(X_{T^k_{j-1}},X_{T^k_j})$ and $K^k_j(\ell) = 0$ otherwise.
\end{lemma}
\begin{proof}
	The proof is given in \textsc{Jeanblanc \& Valchev} \cite[Lemma 2]{jeanblanc-valchev}.
\end{proof}	 
The next theorem shows that the conditional survival probability given the information of a D-investor can be formulated in terms of $\mbb F^D$-conditional survival probabilities for which explicit formulas are derived.
\begin{theorem}\label{th:surv_prob_d-inv}
		Under the assumptions of Theorem \ref{theo_C_investor} for the conditional survival probability given the information of a D-investor it holds
	\begin{align}\label{eq:enl_fil_d-inv}
	\mbb P(\tau > T|\cal G^D_t) = \1{\tau >t }\frac{\mbb P(\tau > T|\cal F^D_t)}{\mbb P(\tau > t|\cal F^D_t)}\qquad \text{for } t < T,
	\end{align}
	where the $\mbb F^D$-conditional survival probabilities are given by the following formulas for $n\geq 2$:\\
	\begin{enumerate}
		\item For $t\in [T_i^k,T_{i+1}^k)$, $k=0,\ldots,n-2$, $i=0,\ldots,\J_k-1$, it holds
		\begin{align}
		\begin{split}\label{eq:sp_d_a}
		\mbb P(\tau >T|\cal F_t^D) =& \int_0^1\int_0^\infty \ldots \int_0^1\int_0^\infty f^{ M,X}_{T_0^{k+1}-T_i^{k}}(u_{k+1},v_{k+1}) \prod\limits_{j=k+2}^{n-1}f^{M,X}_{T^j_0-T^{j-1}_0}(u_j,v_j)\int_0^1f^{ M}_{T-T_0^{n-1}}(w)\\
		&\int_0^{\infty}\int_0^{\infty}\ldots\int_0^{\infty}\int_0^{u_{k+1}X_{T_i^k}}\int_0^{u_{k+2}v_{k+1}X_{T_i^k}}\ldots \int_0^{u_{n-1}v_{n-2}\ldots v_{k+1}X_{T_i^k}}\int_0^{w v_{n-1}\ldots v_{k+1} X_{T_i^k}}\prod_{j=0}^{k-1}\\
		&K^{j,\J_j}(\ell^{j+1})K^{k,i}(\ell^{k+1})f_{L^1,\ldots,L^n}(\ell^1,\ldots,\ell^n)\d \ell^n\ldots\d \ell^1 \d w \d v_{n-1}\d u_{n-1}\ldots \d v_{k+1}\d u_{k+1}.
		\end{split}
		\end{align}
		For $t\in[T_{i}^{n-1},T_{i+1}^{n-1})$, $i=0,\ldots,J_{n-1}-1$, it holds
		\begin{align}
		\begin{split}\label{eq:sp_d_b}
		\mbb P(\tau > T|\cal F^D_t) =& \int_0^\infty \ldots \int_0^\infty K^{n-1,i}(\ell^n)\prod_{j=0}^{n-2}K^{j,\J_j}(\ell^{j+1})\Psi\bigg(T-T_{i}^{n-1},\frac{\ell^n}{X_{T_{i}^{n-1}}} \bigg)f_{L^1,\ldots,L^n}(\ell^1,\ldots,\ell^n)\\
		& \d\ell^n \ldots\d \ell^1.
		\end{split}
		\end{align}
		\item For $t\in [T_i^k,T_{i+1}^k)$, $k=0,\ldots,n-1$, $i=0,\ldots,\J_k-1$, it holds
		\begin{align}\begin{split}\label{eq:sp_d_c}
		\mbb P(\tau > t | \cal F^D_t)= \int_0^\infty \ldots \int_0^\infty &\prod_{j=0}^{k-1}K^{j,\J_j}(\ell^{j+1})K^{k,i}(\ell^{k+1})\Psi\left(t-T_i^k,\frac{\ell^{k+1}}{X_{T_i^k}} \right)\\
		& f_{L^1,\ldots,L^{k+1}}(\ell^1,\ldots,\ell^{k+1})\d \ell^{k+1}\ldots\d \ell^1.
		\end{split}\end{align}
	\end{enumerate}
\end{theorem}
\begin{proof}
	The proof is presented in Appendix \ref{sec:app}.
\end{proof}
\begin{remark}
	For the special case of a constant but random default barrier, i.e., $n=1$ and $L_t=L^1$, the $\mbb F^D$-conditional survival probabilities are given by the following formulas: Let $T_j$, $j=0,\ldots,J-1$, where $0=T_0< T_1<\ldots<T_{\J-1}< T_J=T$ for $\J \in \mbb N$, denote the times where the D-investor obtains information about the firm's asset value. For $t\in [T_i,T_{i+1})$, $i=0,\ldots,J-1$, it holds
	\begin{align}\label{eq:cond_prob_d_giesecke}
	\mbb P(\tau > T|\cal G^D_t)=& \1{\tau >t }\frac{\int_0^1 \int_0^{X_{T_i}w} \prod_{j=1}^{i}K_j(\ell^{1}) f_{L^1}(\ell^{1})f^{\widehat M}_{T-T_i}(w)\d \ell^{1} \d w}{	\mbb P(\tau > t|\cal F^D_t)=\int_0^1 \int_0^{X_{T_i}w} \prod_{j=1}^{i}K_j(\ell^{1}) f_{L^1}(\ell^{1})f^{\widehat M}_{t-T_i}(w)\d \ell^{1} \d w},
	\end{align}
	where
	\begin{align*}
	K_j(\ell)=\mbb P(M_{[T_{j-1},T_j)}>\ell| X_{T_{j-1}},X_{T_j}).
	\end{align*}
\end{remark}
\begin{remark}
	Let $C_t$ be the price of a zero-coupon bond that matures at $T$ and has zero recovery. Assuming that the pricing probability is $\mbb P$, then $C_t$ is given by
	\begin{align*}
	C_t = e^{-r(T-t)}\mbb E[\1{\tau >T}\,|\,\cal G^D_t]
	=e^{-r(T-t)} \1{\tau >t }\frac{\mbb P(\tau > T|\cal F^D_t)}{\mbb P(\tau > t|\cal F^D_t)},
	\end{align*}
	where $r\geq 0$ is a discount factor. The credit spread $S_t$ is given by
	\begin{align*}
	S_t=-\frac{1}{T-t}\ln\{\mbb P(\tau > T|\cal G^D_t)\}.
	\end{align*}
\end{remark}	

\section{Numerical examples}\label{sec:numeric}
In this section we implement the formulas for the default probabilities derived in the previous section by evaluating the integrals using the Gauss–Kronrod quadrature formula (see \textsc{Monegato} \cite{Monegato}) and quantify numerically the impact of asymmetric information on the estimations of the default probabilities and credit spreads. More numerical examples can be found in \textsc{Redeker} \cite{redeker}.\\
The parameters for the firm's asset value process are taken from \textsc{Blanchet-Scalliet, Hillairet \& Jiao} \cite{blanchet-hillairet}, i.e.,
\begin{align*}
X_0=1, \quad \mu=0.05 \quad \text{and} \quad \sigma=0.8.
\end{align*}
We consider a time horizon of $T=2$ years and we suppose that a firm's management decides at $t_0=0$ on a default threshold $L^1$. Further, we assume that the management adjusts the default threshold at $t_1=1$ from $L^1$ to $L^2$. C-investors and D-investors only have knowledge on the (marginal and joint) laws of $\mathbf{L}=(L^1, L^2)$. The law of $\mathbf{L}=(L^1, L^2)$ is modeled by a copula. 
In the following examples C-investors and D-investors assume that $L^1$ is beta distributed with parameters $(\alpha,\beta)=(2,2)$ and $L^2$ is exponentially distributed with parameter $\lambda=2/3$. Further, investors assume that the law of $\mathbf{L}$ is given by a Gumbel copula $C$, i.e., 
\begin{align*}
C(x_1,x_2)=\exp\left\{-\left[(-\ln(x_1))^{\theta}+(-\ln(x_2))^{\theta}\right]^{\frac{1}{\theta}}\right\},
\end{align*}
for some $\theta \ge 1$ (see \textsc{Bluhm \& Overbeck} \cite{bluhm}). Thus, the correlation between $L^1$ and $L^2$ is modeled by the parameter $\theta$, where $\theta=1$ corresponds to the case of independent default thresholds.\\
The top panel of Figure \ref{fig:firm-value-nodefault} presents a realized trajectory of the firm's asset process, the switching default threshold and the running minimum of the firm value which is restarted after adjustment of the default threshold. We observe that no default occurs before maturity and the realizations of the default thresholds are given by $L^1(\omega) = 0.6$ and $L^2(\omega)=1.2$. The middle panel of Figure \ref{fig:firm-value-nodefault} shows the associated conditional survival probability given the information of a C-investor for different values of $\theta$ ($\theta=1, 2, 100$). We observe that the conditional survival probability converges to one, since no default has occurred by maturity. If $L^1$ and $L^2$ are independent, i.e., $\theta=1$, the estimate of the survival probability is smaller compared to the cases of dependent $L^1$ and $L^2$. However, differences in the conditional survival probability between $\theta=1$, $\theta=2$ and $\theta=100$ become smaller with increasing time. The bottom panel of Figure \ref{fig:firm-value-nodefault} illustrates the associated credit yield spread given the information of a C-investor for different values of $\theta$ ($\theta=1, 2, 100$). The credit spread tends to zero as the time $t$ approaches maturity $T$ since the C-investor has learned about the default threshold and knows that the firm is not subject to default in the next instance of time. If $L^1$ and $L^2$ are independent, i.e., $\theta=1$, the credit yield spread is higher compared to the cases of dependent $L^1$ and $L^2$.
\begin{figure}[h!]
	\centering
	\captionsetup{format=hang, justification=centerfirst, textfont=normalsize, labelfont=normalsize}
	\begin{subfigure}[b]{1.2\textwidth}
		\hspace*{-0.1\textwidth}
		\includegraphics[width=1\textwidth,height=0.23\textheight]{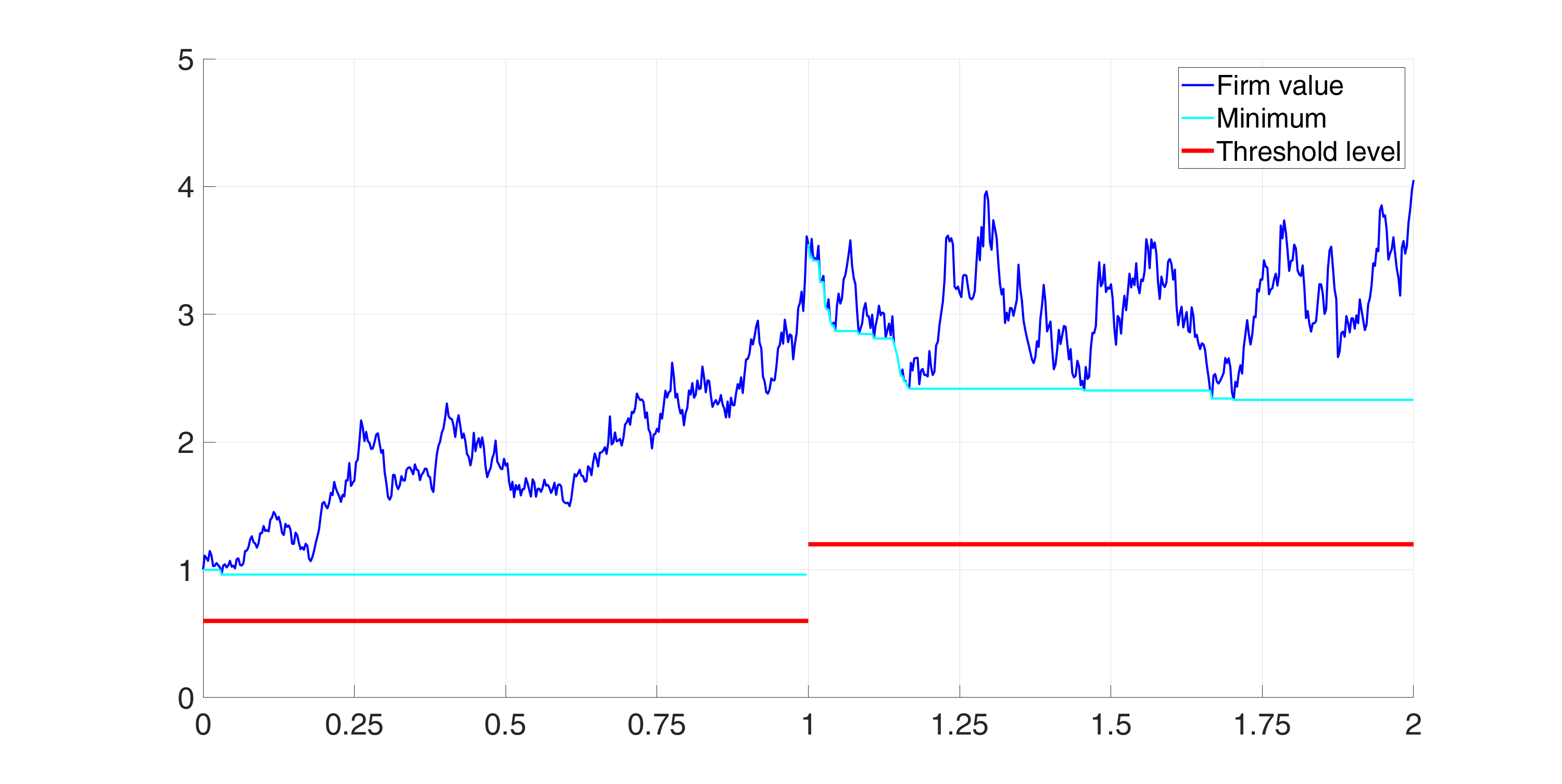}	
		\caption{Firm value, running minimum process and default threshold.}
	\end{subfigure}
	\begin{subfigure}[b]{1.2\textwidth}  
		\hspace*{-0.1\textwidth}
		\includegraphics[width=1\textwidth,height=0.23\textheight]{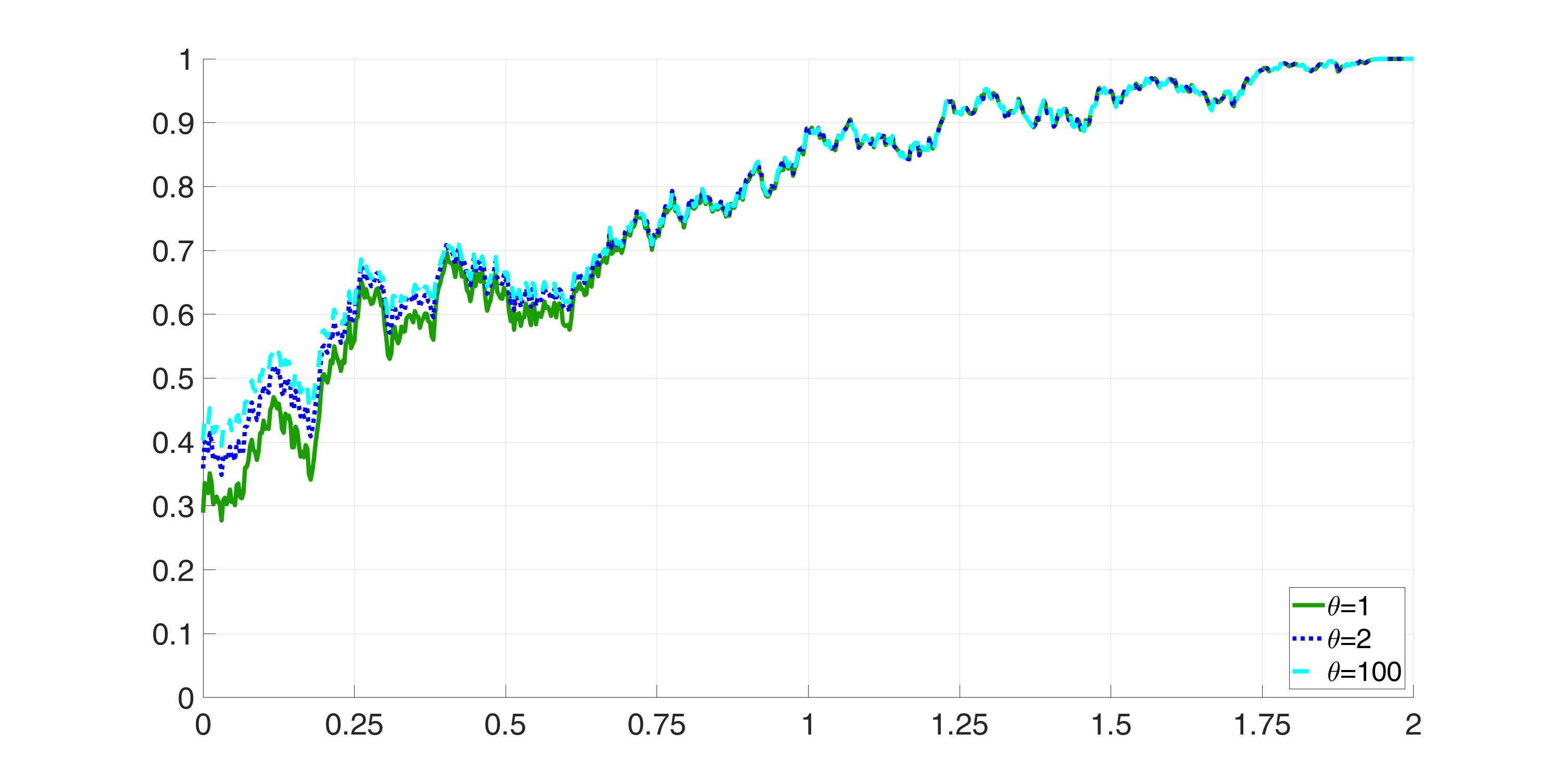}  
		\caption{Conditional survival probability given the information of a C-investor.}
	\end{subfigure}
	\begin{subfigure}[b]{1.2\textwidth}  
		\hspace*{-0.1\textwidth}\includegraphics[width=1\textwidth,height=0.23\textheight]{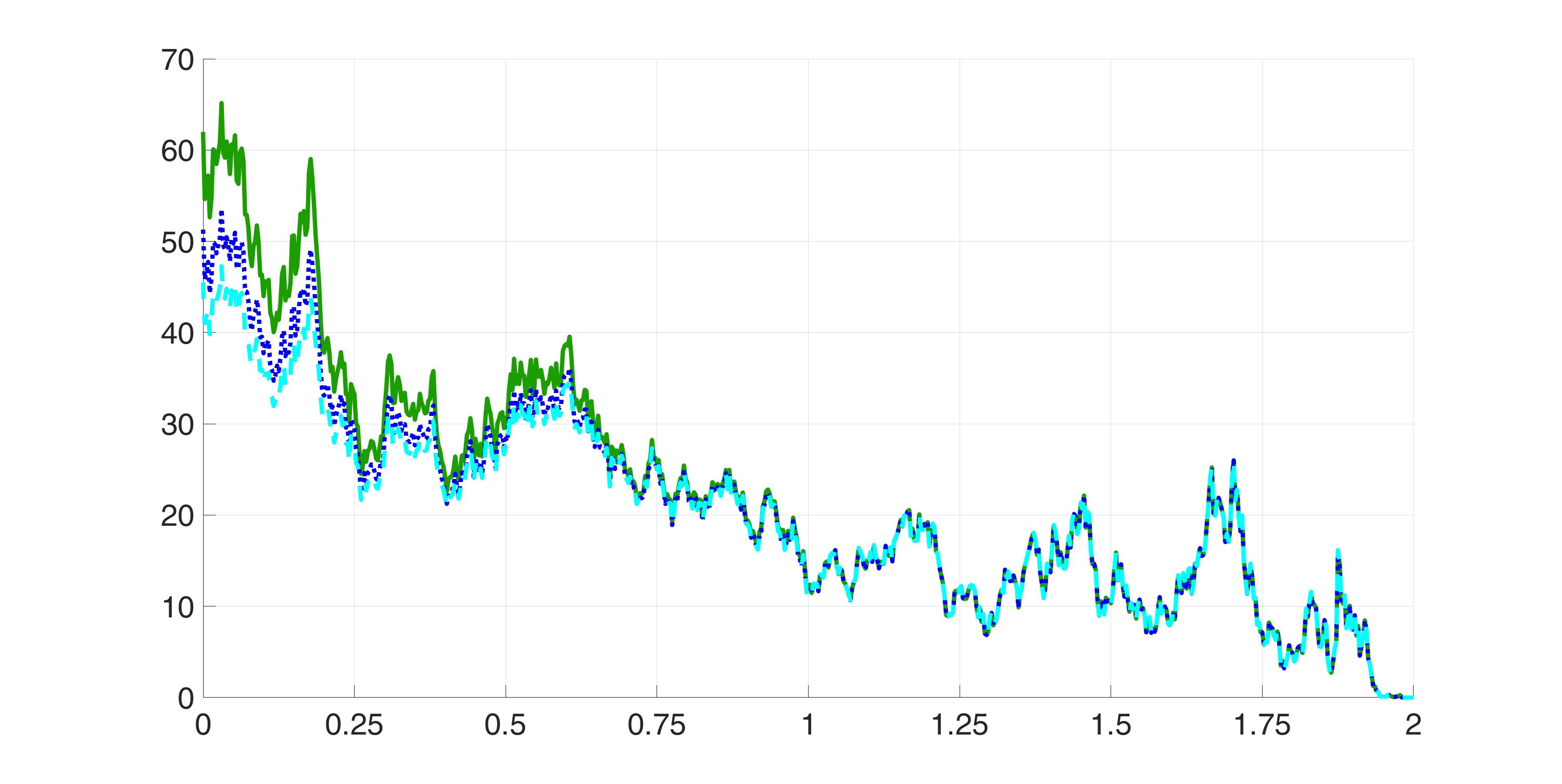}
		\caption{Credit yield spread given the information of a C-investor.}
	\end{subfigure}
   \captionsetup{format=hang, justification=raggedright,textfont=normalsize, labelfont=normalsize}
	\caption[Switching default thresholds - no default case: Firm value and the associated conditional survival probability and credit spread given the information of a C-investor for different values of $\theta$.]{Plot of a trajectory of the firm's asset process and the associated conditional survival probability and credit spread given the information of a C-investor for different values of $\theta$.}\label{fig:firm-value-nodefault}
\end{figure}\\
Figure \ref{fig:firm-value-nodefault-D} illustrates the conditional survival probability and the credit spread given the information of a D-investor for the case of independent default thresholds ($\theta=1$) and different information time points (biannually, quarterly, monthly). 
\begin{figure}[ht!]
	\centering
	\captionsetup{format=hang, justification=centerfirst, textfont=normalsize, labelfont=normalsize}
	\begin{subfigure}[b]{1.2\textwidth}		
		\hspace*{-0.1\textwidth}
		\includegraphics[width=1.0\textwidth,height=0.23\textheight]{Firm_Value_noDefault_1.png}
		\caption{Firm value, running minimum process and default threshold.}
	\end{subfigure}
	\begin{subfigure}[b]{1.2\textwidth}
		\hspace*{-0.1\textwidth}
		\includegraphics[width=1.0\textwidth,height=0.23\textheight]{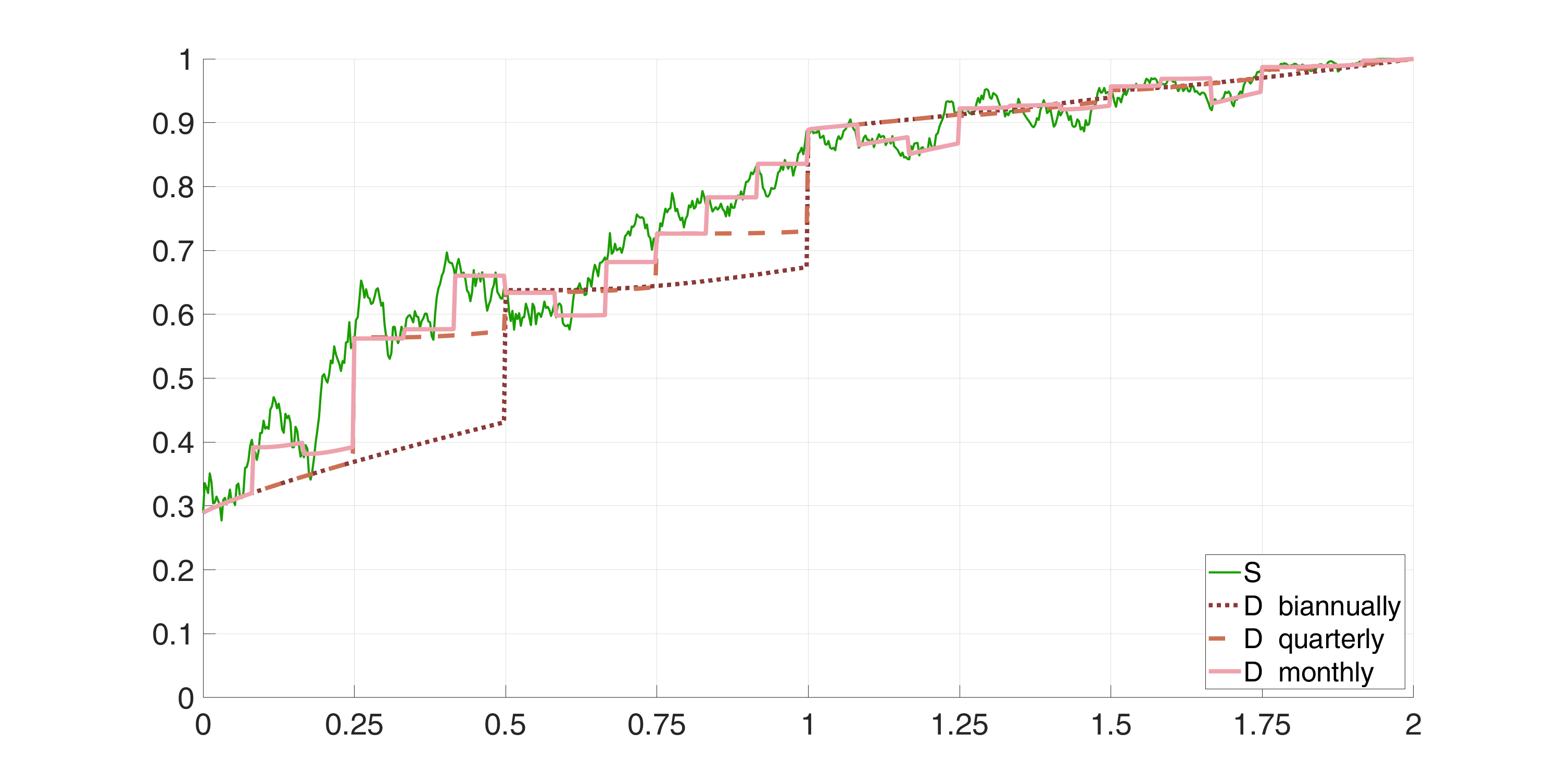}
		\caption{Conditional survival probabilities given the information of a C-investor and D-investors.}
	\end{subfigure}
	\begin{subfigure}[b]{1.2\textwidth}
		\hspace*{-0.1\textwidth}
		\includegraphics[width=1.0\textwidth,height=0.23\textheight]{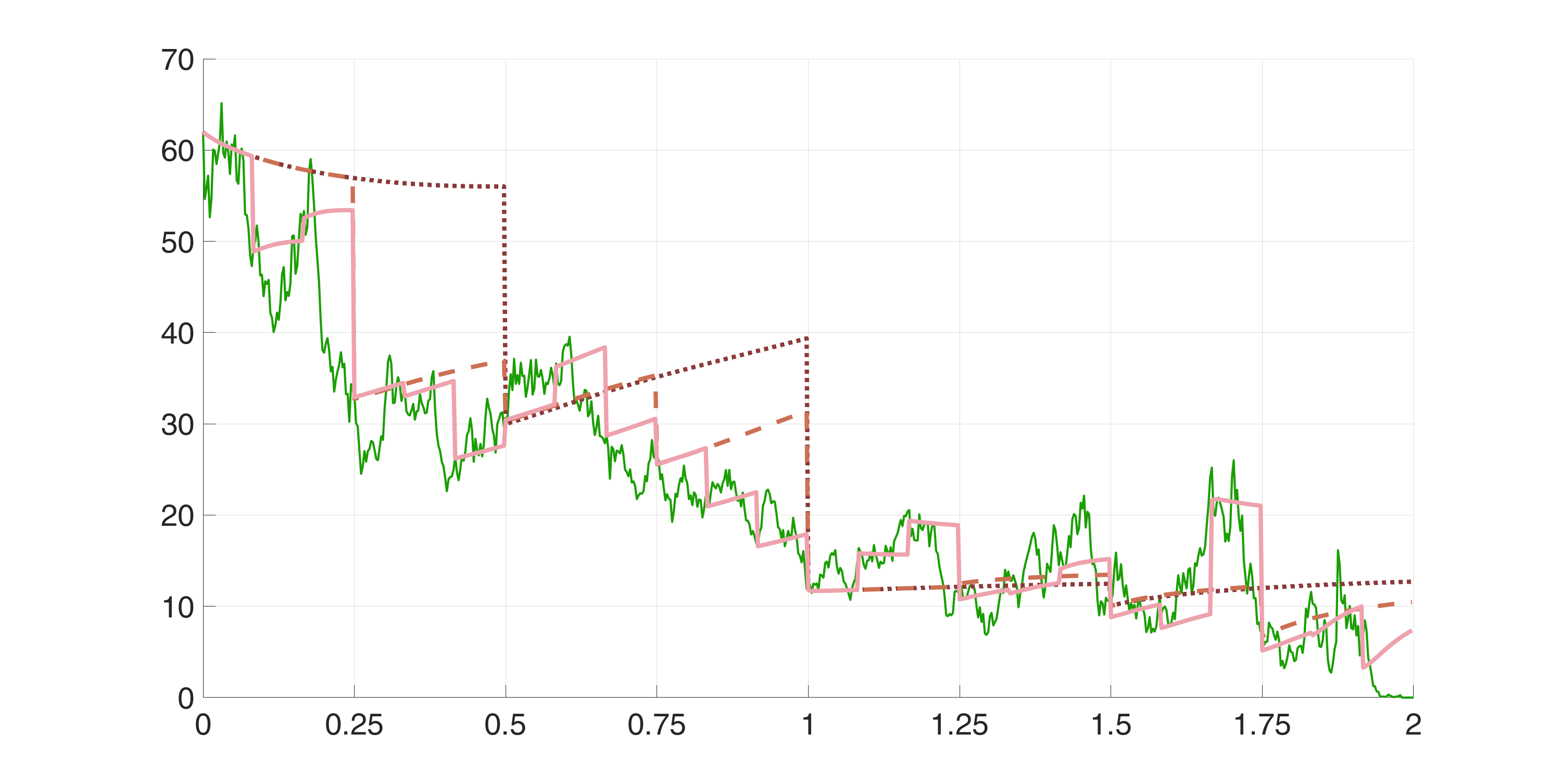}
		\caption{Credit yield spreads given the information of a C-investor and D-investors.}
	\end{subfigure}	
	\captionsetup{format=hang, justification=raggedright,textfont=normalsize, labelfont=normalsize}
	\caption[Switching default thresholds - no default case: Firm value and the associated conditional survival probabilities and credit spreads given the information of a C-investor and D-investors.]{Plot of a trajectory of the firm's asset process and the associated conditional survival probabilities and credit spreads given the information of a C-investor and D-investors for $\theta=1$.}\label{fig:firm-value-nodefault-D}
\end{figure}
We observe that the more frequently a D-investor obtains information about the firm value the closer are the D-investor's and C-investor's estimates of the survival probability. The credit yield spread given the information of a D-investor is non-zero at maturity, i.e., D-investors demand a risk premium for the default risk. D-investors who obtain information about the firm value twice a year demand the highest risk premium followed by D-investors who obtain that information  every quarter. The lowest risk premium is demanded by D-investors who obtain information about the firm value every month. This indicates that the risk premium depends on the frequency of observations of the firm value and the associated time to maturity at the last information date.  
Recall that the D-investors who obtain information about the firm value twice a year, every quarter and every month receive the last information before $T$ at $t=1.5$, $t=1.75$ and $t=1.9167$, respectively. At these time points the firm value is $X_{1.5} \approx 3.24$,  $X_{1.75} \approx 3.27$ and  $X_{1.9167} \approx 2.96$. Thus, the D-investor who observes the firm value every month demands the highest risk premium. Furthermore, the last observed firm value before $T$ is almost equal for the other two D-investors but the D-investor who obtains information twice a year demands a higher risk premium compared to the D-investor who observes  the firm value every quarter. This indicates that the risk premium also depends on the frequency of observations of the firm value and the associated time to maturity at the last information date.
\begin{figure}[h!]
	\centering
	\captionsetup{format=hang, justification=centerfirst, textfont=normalsize, labelfont=normalsize}
	\begin{subfigure}[b]{1.2\textwidth}		
		\hspace*{-0.1\textwidth}
		\includegraphics[width=1\textwidth,height=0.23\textheight]{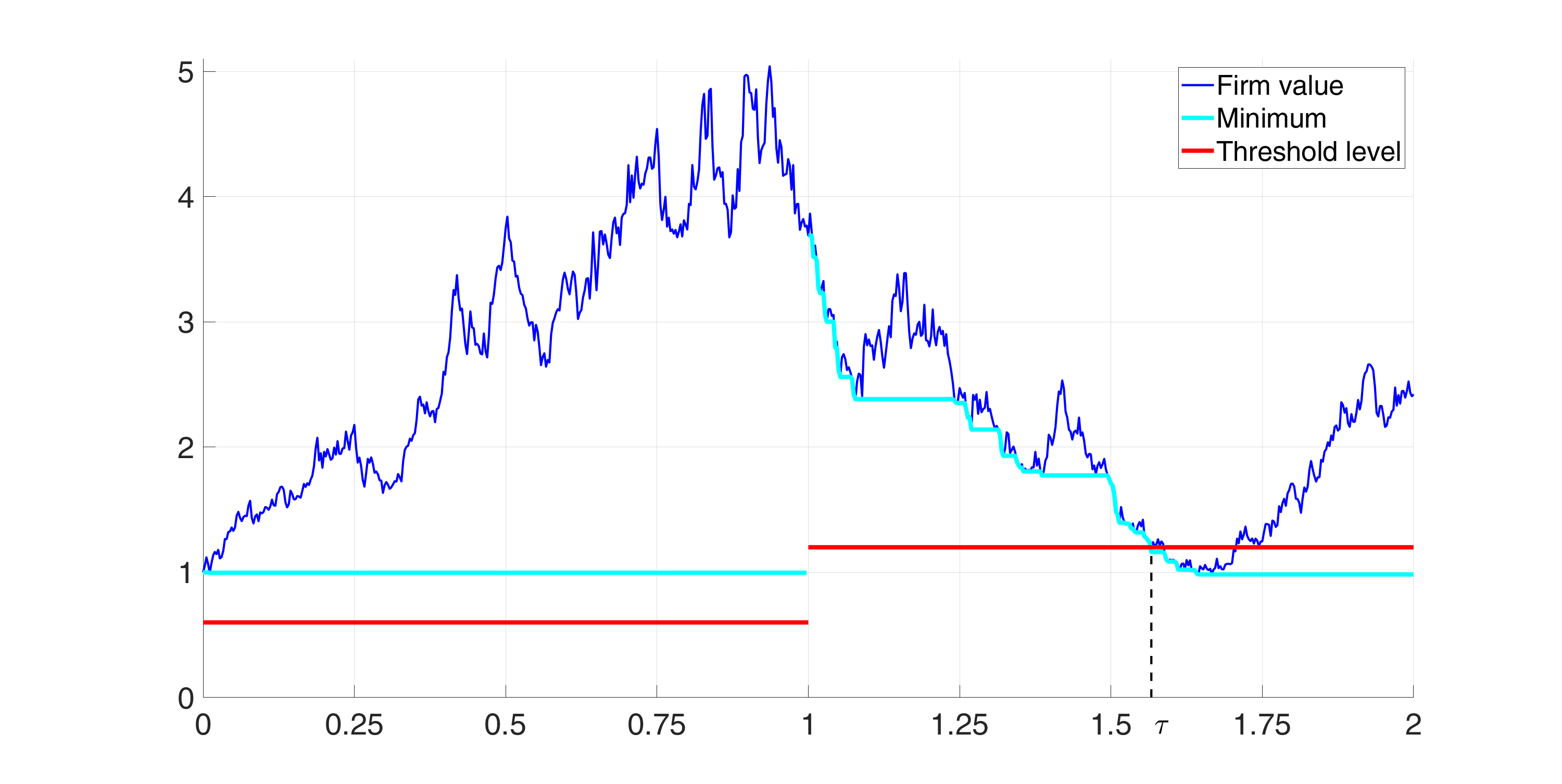}	
		\caption{Firm value, running minimum process and default threshold.}
	\end{subfigure}
	\begin{subfigure}[b]{1.2\textwidth}  
		\hspace*{-0.1\textwidth}
		\includegraphics[width=1\textwidth,height=0.23\textheight]{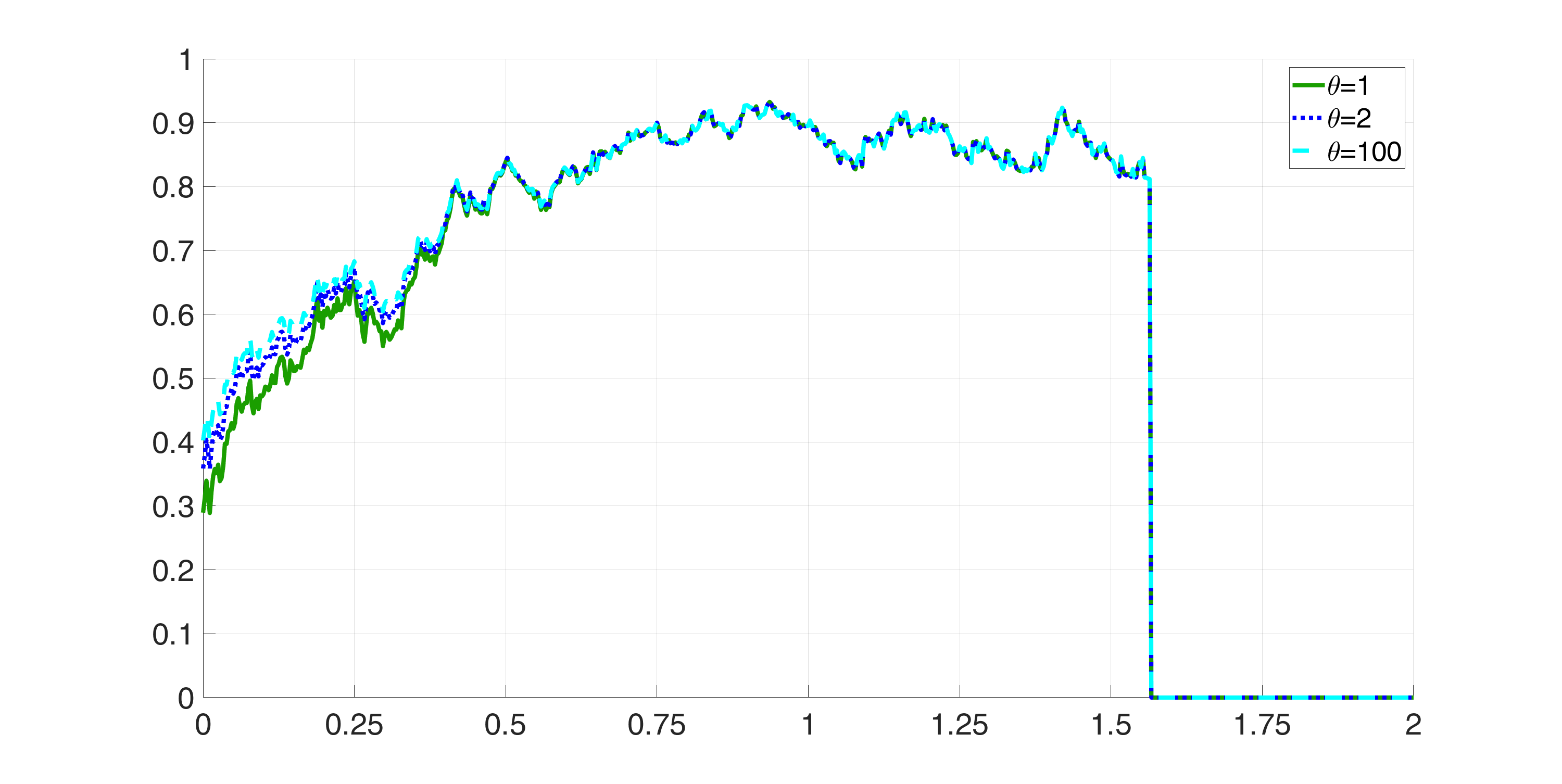}  
		\caption{Conditional survival probability given the information of a C-investor.}
	\end{subfigure}
	\begin{subfigure}[b]{1.2\textwidth}  
		\hspace*{-0.1\textwidth}\includegraphics[width=1\textwidth,height=0.23\textheight]{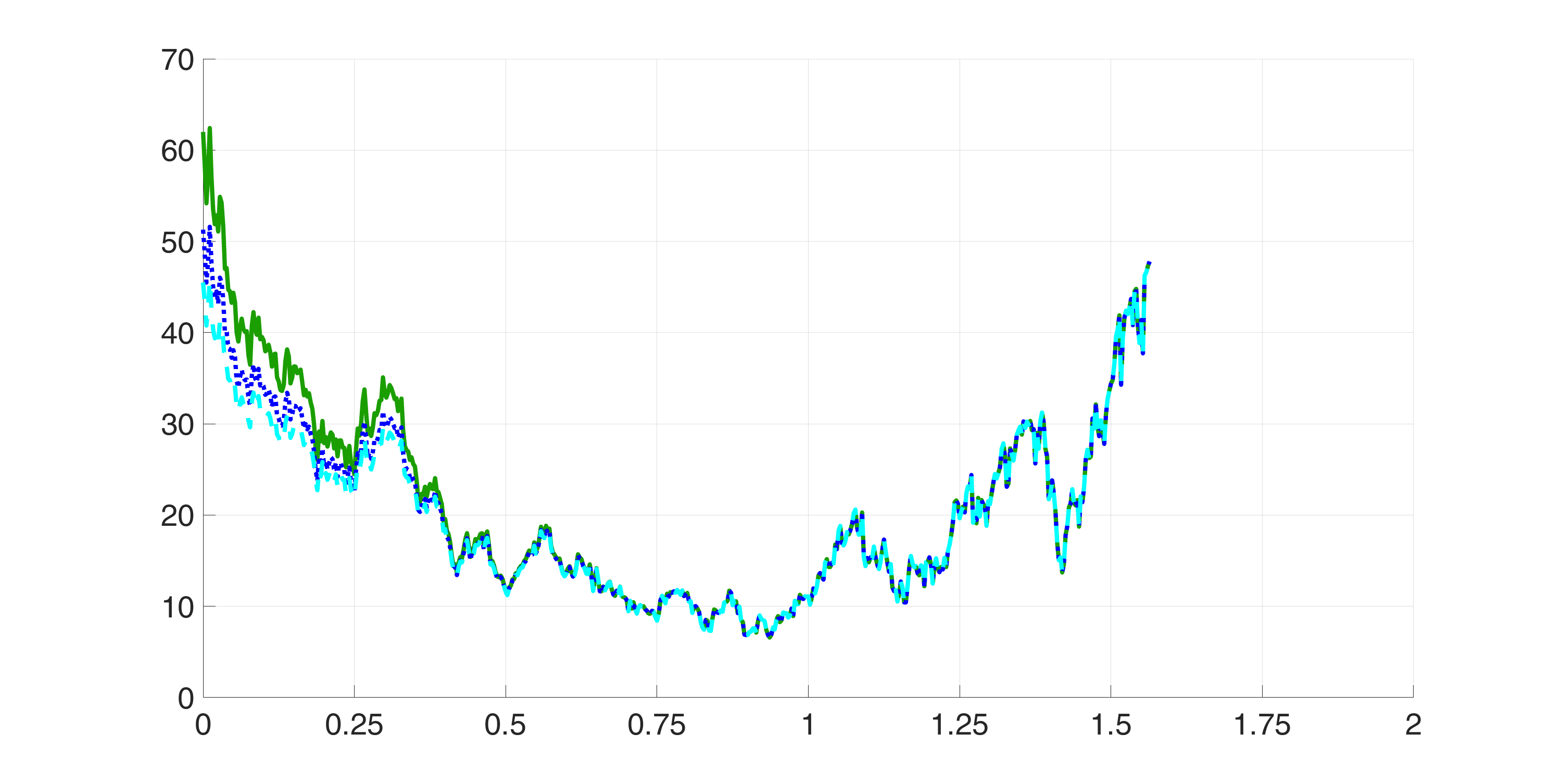}
		\caption{Credit yield spread given the information of a C-investor.}
	\end{subfigure}
	\captionsetup{format=hang, justification=raggedright,textfont=normalsize, labelfont=normalsize}
	\caption[Switching default thresholds - default case: Firm value and the associated conditional survival probability and credit spread given the information of a C-investor for different values of $\theta$.]{Plot of a trajectory of the firm's asset process and the associated conditional survival probability and credit spread given the information of a C-investor for different values of $\theta$.}\label{fig:firm-value-default}
\end{figure}\\
A second example is presented by Figure \ref{fig:firm-value-default}. The top panel shows a realized trajectory of the firm's asset process, the switching default threshold and the running minimum of the firm value which is restarted after adjustment of the default threshold. We observe that default occurs in the second year. The middle and bottom panel illustrate the associated conditional survival probability and credit yield spread given the information of a C-investor for different values of $\theta$ ($\theta=1, 2, 100$). We observe that the conditional survival probability jumps to zero at the time of default. 
\begin{figure}[ht!]
	\centering
	\captionsetup{format=hang, justification=centerfirst, textfont=normalsize, labelfont=normalsize}
	\begin{subfigure}[b]{1.2\textwidth}		
		\hspace*{-0.1\textwidth}
		\includegraphics[width=1.0\textwidth,height=0.23\textheight]{Firm_Value_Default_1.png}
		\caption{Firm value, running minimum process and default threshold.}
	\end{subfigure}
	\begin{subfigure}[b]{1.2\textwidth}
		\hspace*{-0.1\textwidth}
		\includegraphics[width=1.0\textwidth,height=0.23\textheight]{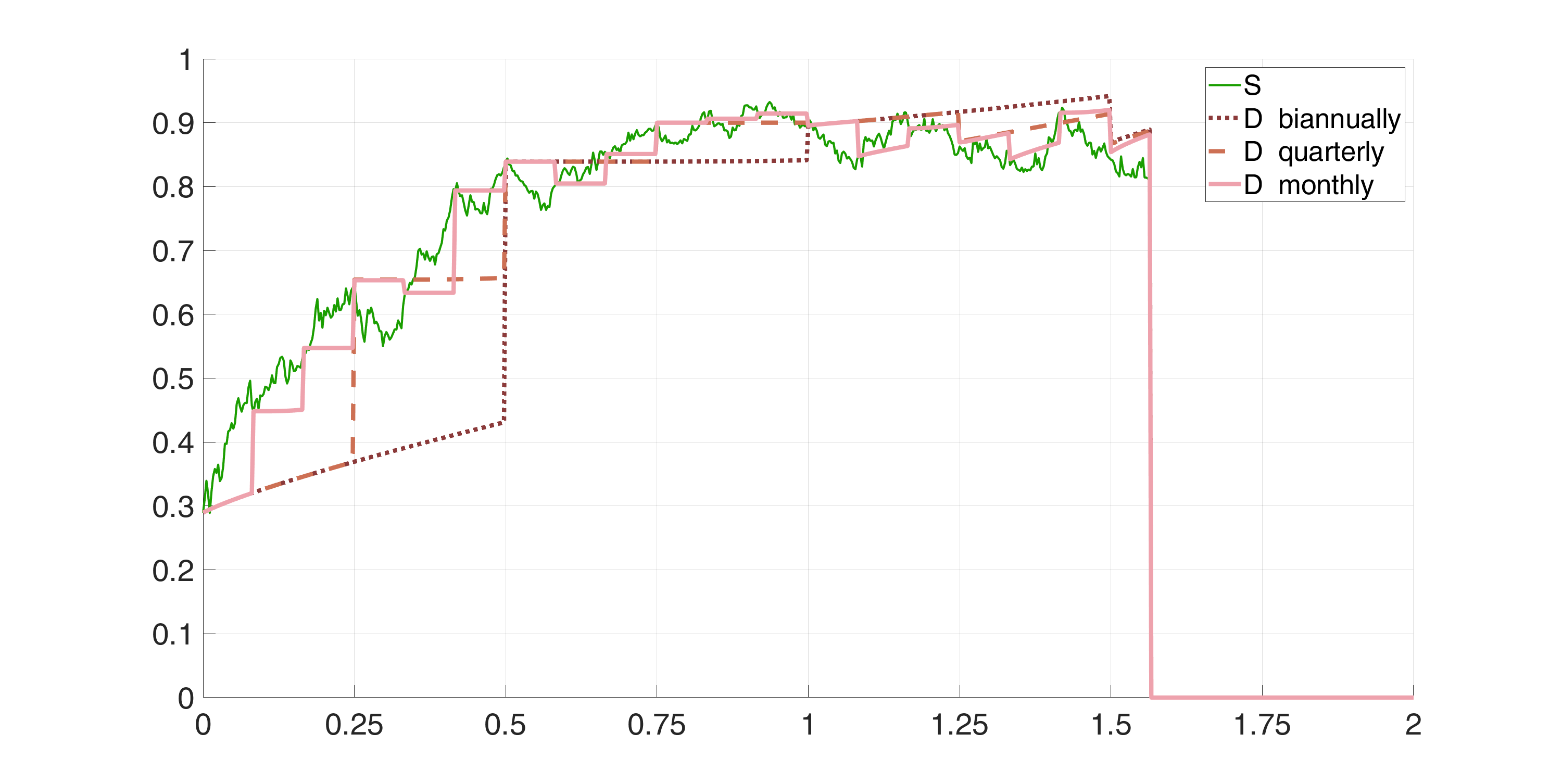}
		\caption{Conditional survival probabilities given the information of a C-investor and D-investors.}
	\end{subfigure}
	\begin{subfigure}[b]{1.2\textwidth}
		\hspace*{-0.1\textwidth}
		\includegraphics[width=1.0\textwidth,height=0.23\textheight]{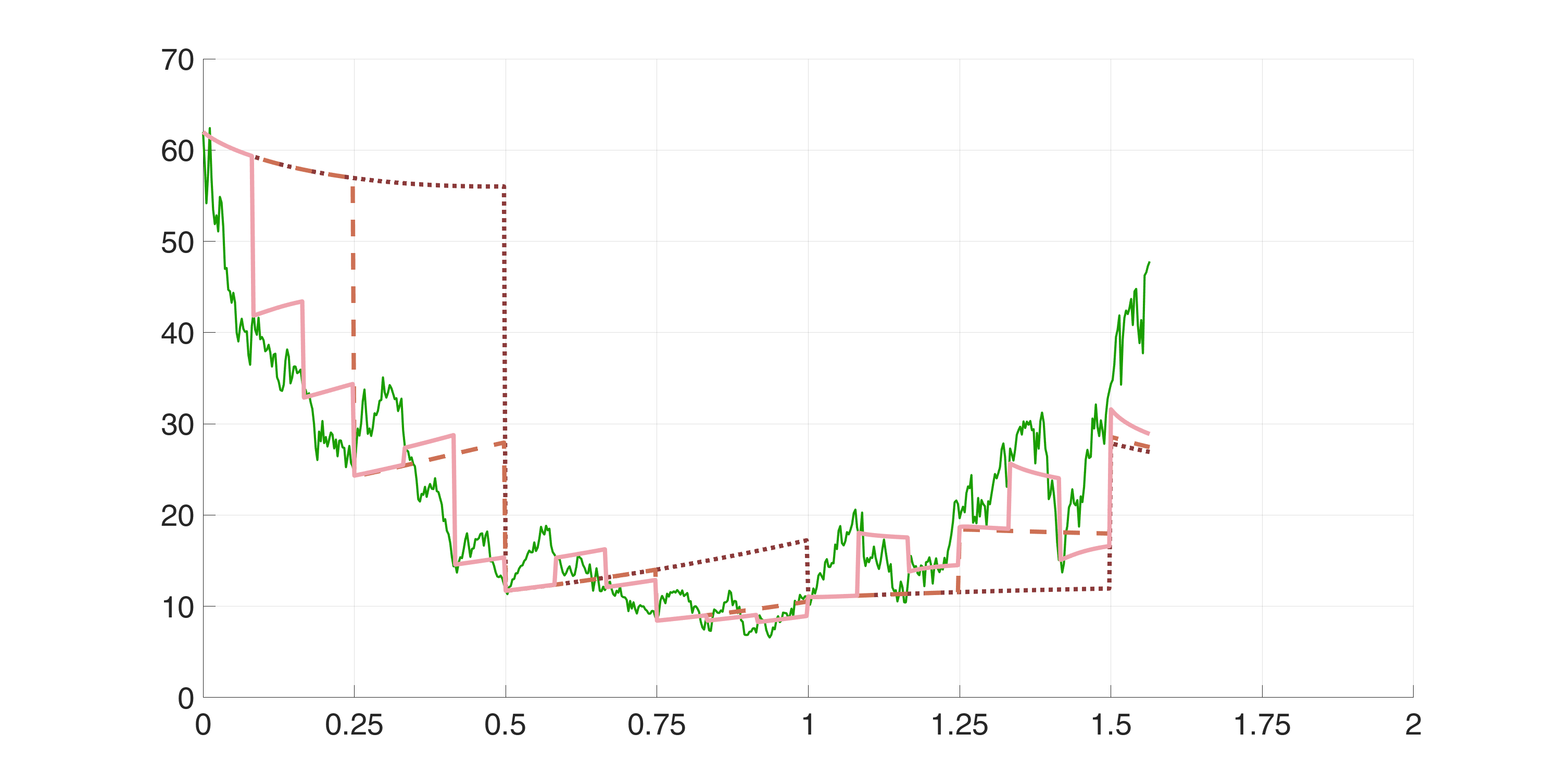}
		\caption{Credit yield spreads given the information of a C-investor and D-investors.}
	\end{subfigure}	
	\captionsetup{format=hang, justification=raggedright,textfont=normalsize, labelfont=normalsize}
	\caption[Switching default thresholds - default case: Firm value and the associated conditional survival probabilities and credit spreads given the information of a C-investor and D-investors.]{Plot of a trajectory of the firm's asset process and the associated conditional survival probabilities and credit spreads given the information of a C-investor and D-investors for $\theta=1$.}\label{fig:firm-value-default-D}
\end{figure}\\
Figure \ref{fig:firm-value-default-D} illustrates the conditional survival probability and credit yield spread given the information of a D-investor for the case of independent default thresholds ($\theta=1$). The D-investors obtain information about the firm value twice a year, every quarter and every month, respectively. We observe that the conditional survival probabilities jump to zero at the time of default.\\
\begin{figure}[ht!]
	\centering
	\captionsetup{format=hang, justification=centerfirst, textfont=normalsize, labelfont=normalsize}
	\begin{subfigure}[b]{1.2\textwidth}		
		\hspace*{-0.1\textwidth}
		\includegraphics[width=1.0\textwidth,height=0.23\textheight]{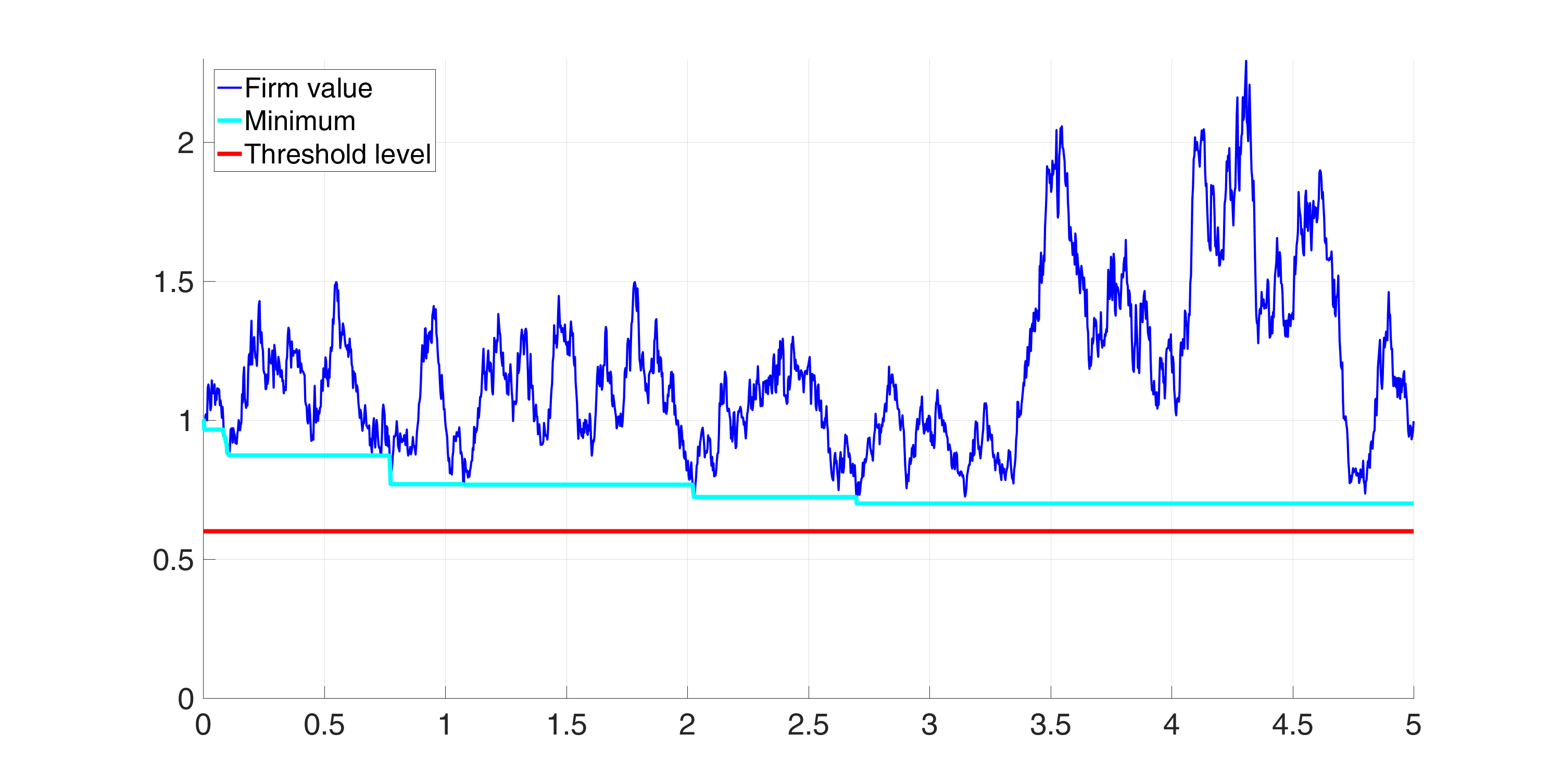}
		\caption{Firm value, running minimum process and default threshold.}
	\end{subfigure}
	\begin{subfigure}[b]{1.2\textwidth}
		\hspace*{-0.1\textwidth}
		\includegraphics[width=1.0\textwidth,height=0.23\textheight]{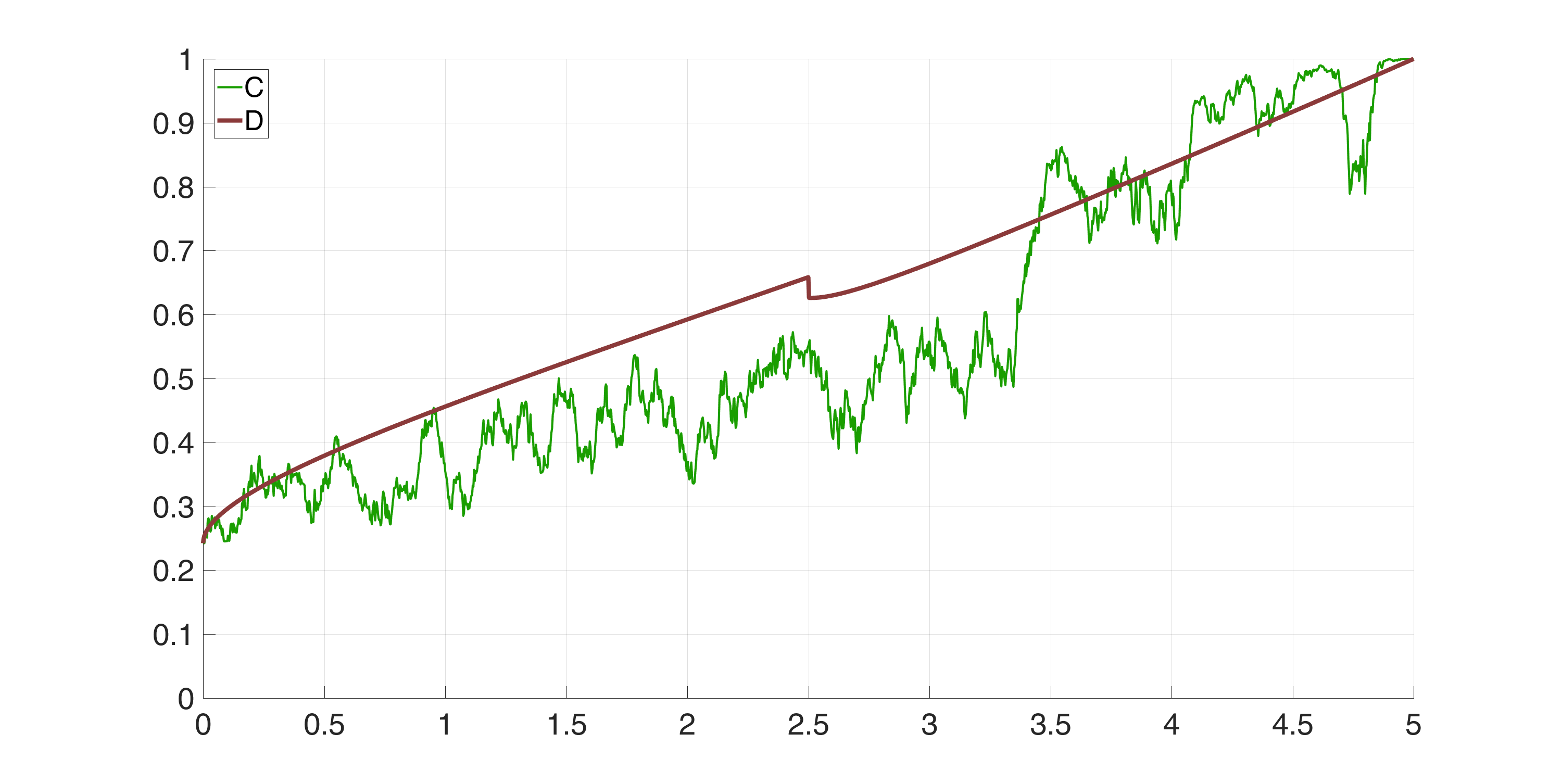}
		\caption{Conditional survival probabilities given the information of a C-investor and a D-investor.}
	\end{subfigure}
	\begin{subfigure}[b]{1.2\textwidth}
		\hspace*{-0.1\textwidth}
		\includegraphics[width=1.0\textwidth,height=0.23\textheight]{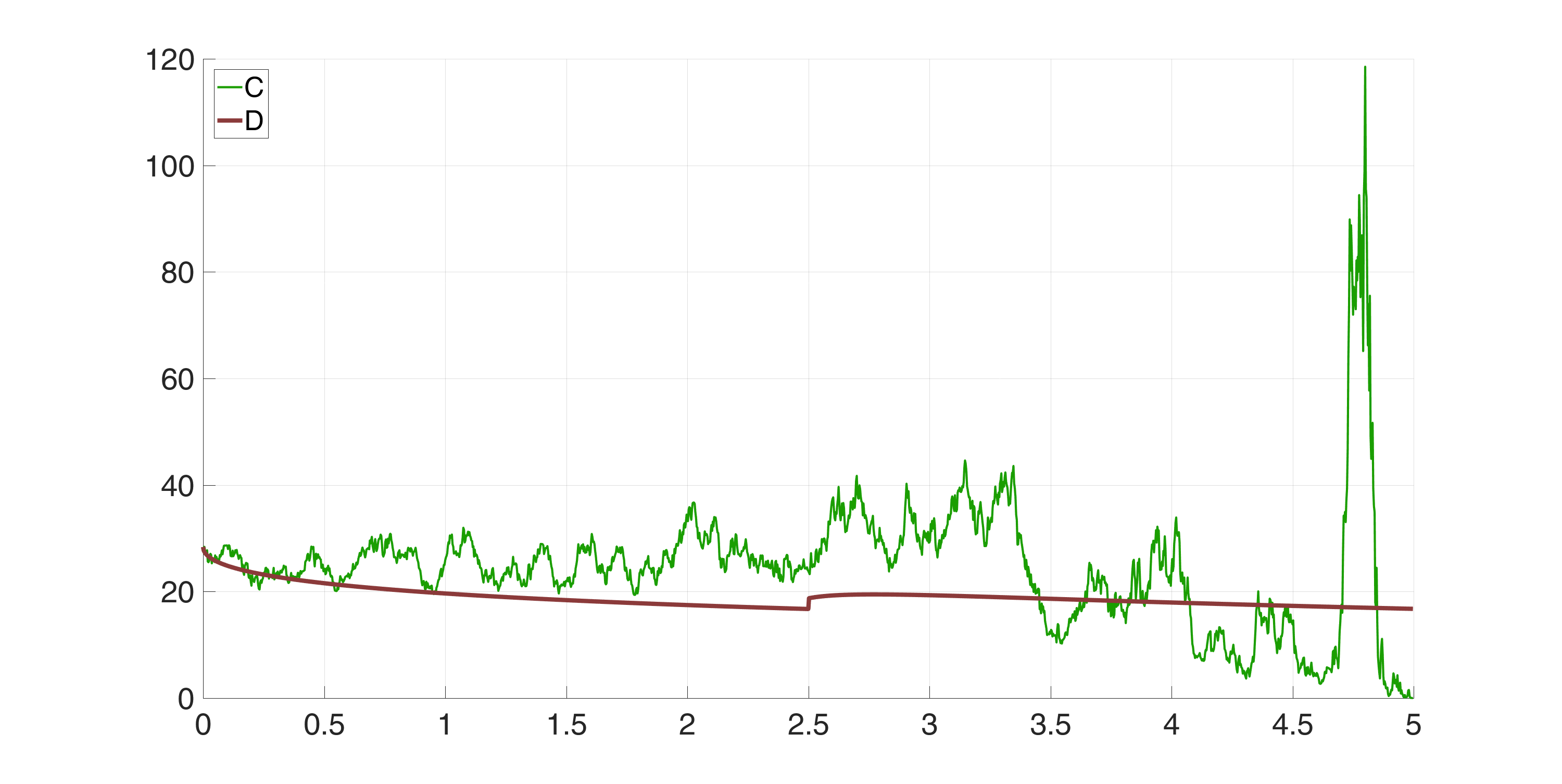}
		\caption{Credit yield spreads given the information of a C-investorand a D-investor.}
	\end{subfigure}	
	\captionsetup{format=hang, justification=raggedright,textfont=normalsize, labelfont=normalsize}
	\caption{Plot of a trajectory of the firm's asset process and the associated conditional survival probabilities and credit yield spreads given the information of a C-investor and a D-investor. The D-investor obtains information about the firm value once in 2.5 years.}\label{fig:surv-prob-L20}
\end{figure}
We note that the estimates of the conditional survival probabilities given the information of a C-investor and a D-investor are very close at the D-investor's information dates in the examples presented above. This is no longer the case if we extend the time horizon. We present a special case, where the default threshold is taken to be just a random constant of the form $y = L$ as proposed in \textsc{Giesecke \& Goldberg} \cite{giesecke-goldberg} and the time horizon is extended to $T=5$ years. C-investors and D-investors assume that $L$ follows a standard uniform distribution. The top panel of Figure \ref{fig:surv-prob-L20} shows a realized trajectory of the firm value, the associated running minimum and a realization of the default threshold. The middle and bottom panel show the associated conditional survival probabilities and credit yield spreads given the information of a C-investor and a D-investor, respectively. The D-investor obtains information about the firm value once in 2.5 years. We observe that the C-investor's estimate of the conditional survival probability and the D-investor's estimate of the conditional survival probability are now visibly different at the information date of the D-investor. The credit yield spread at $T$ is zero given the information of the C-investor and nonzero given the information of the D-investor.

\section{Conclusion}\label{sec:conclusio}
This paper extends the traditional structural model for credit risk. In the proposed model the default of a firm is triggered when the value of the firm's assets falls below a threshold which is modeled as a sequence of  random variables whose values are chosen by the management of the firm and dynamically adjusted accounting for changes in the economy or the appointment of a new firm management. Investors on the market have no access to the value of the threshold and only anticipate its distribution. Different information levels on the firm value are distinguished and explicit formulas for the conditional default probability and associated credit yield spreads given these information levels are derived. Numerical illustrations are provided which show that the information level has a considerable impact on the estimation of the default probability and the associated credit yield spread. Investors who have perfect information on the value process of the firm learn about the default threshold, i.e., they learn that the default threshold must lie below the current running minimum of the firm value if default has not yet occurred. Thus, the larger the distance of the firm value to the running minimum the less likely is a default and investors adjust their estimation of the default probability accordingly. The associated credit yield spreads are high if the firm value is close to its running minimum and low otherwise. Especially, if the firm value is far above its running minimum just before maturity investors know that there will be no default in the next instance of time and they do not demand a default risk premium, i.e., the credit spread is zero. This is different for investors who do not have full access to the value process of the firm. Investors who only observe the firm value at specific dates cannot be certain about the firm value just before maturity and they demand a nonzero default risk premium. Furthermore, the credit spreads at maturity depend on their last observed firm value and the frequency of firm value observations. In future research the dynamics of defaultable bonds are studied in this model set-up.
\appendix
\section{Proof of Theorem \ref{th:surv_prob_d-inv}}\label{sec:app}
Observe that $\mbb F^D \subset \mbb F$. Then the proof of Eq. \eqref{eq:enl_fil_d-inv} is along the same line as the proof of Eq. \eqref{eq:enl_fil} for the case of a C-investor. For the sake of a simple notation we prove the remaining formulas of the theorem for the case $n=2$, i.e., the threshold is $L^1$ in the interval $[t_0,t_1)$ and $L^2$ in the interval $[t_1,T)$. The proof for $n>2$ is along the same line and skipped. In order to keep the notation simple we denote the times at which the D-investor observes the asset process by $T_i$, $i=0,\ldots,J-1$, where $T_0:=t_0$, $T_{J}:=T$ and $T_j=t_1$ for some $j \in \{1,\ldots,\J-1\}$. The last term ensures that the D-investor obtains information about the asset process at the adjustment time $t_1$ of the default threshold. We define the processes $(Y^i_u)_{u\geq0}$ by
\begin{align*}
Y^i_u =  \exp\{mu+\sigma \widehat B^i_{u} \},\qquad \text{for } i=1,\ldots,\J-1,
\end{align*}
where $(\widehat B^i_u)_{u\geq 0}$ is a Brownian motion independent of $\cal F_{T_i}$ and given by $\widehat B^i_u = B_{T_i+u}-B_{T_i}$. Note that $Y^i_u$ inherits the independence of $\cal F_{T_i}$ and it has the same law as $X_u$. Further we have the decomposition $X_s = X_{T_i} Y^i_{s-T_i}$ for $s > T_i$. We denote by $(\widehat M^i_u)_{u\geq 0}$ the running minimum of $Y^i$, i.e.,
\be 
\widehat M^i_u = \inf_{s < u} Y^i_s.
\ee
\myparagraph{Proof of \eqref{eq:sp_d_a}:} 
Let $t\in [t_0,t_1)$ and $t\in [T_i,T_{i+1})$ with $i<j$, i.e., $T_{i+1}\leq T_j =t_1$, then
	\begin{align*}
	\{\tau >T\}=&\{L^1 < M_{t_1}\}\cap\{L^2 < M_{[t_1,T)}\}=\{L^1 < M_{T_i}\}\cap\{L^1 < M_{[T_i,t_1)}\}\cap \{ L^2 < M_{[t_1,T)}\}\\
	=&\{L^1 < M_{T_i}\}\cap\{L^1 < M_{[T_i,T_j)}\}\cap \{ L^2 < M_{[T_j,T)}\}\\
	=&\{L^1 < M_{T_i}\}\cap\{L^1 < \widehat M^i_{T_j-T_i}X_{T_i}\}\cap\{L^2 < \widehat M^j_{T-T_j}Y^i_{T_j-T_i}X_{T_i}\}.
	\end{align*}
	Then we have
	\begin{align*}
	\mbb P(\tau >T|\cal F^D_t) =& \mbb P(\tau >T|\cal F^D_{T_i})=\mbb P(L^1 < M_{T_i}, L^1 < \widehat M^i_{T_j-T_i}X_{T_i}, L^2 < \widehat M^j_{T-T_j}Y^i_{T_j-T_i}X_{T_i}|\cal F^D_{T_i})\\
	=&\int_0^1 \int_{0}^{\infty} \mbb P(L^1 < M_{T_i}, L^1 < uX_{T_i}, L^2 < \widehat M^j_{T-T_j}vX_{T_i}|\cal F^D_{T_i})f^{\widehat M^i,Y^i}_{T_j-T_i}(u,v)\d v \d u\\
	=&\int_0^1\int_{0}^{\infty} \int_0^1 \mbb P(L^1 < M_{T_i}, L^1 < uX_{T_i}, L^2 < wvX_{T_i}|\cal F^D_{T_i})f^{\widehat M^j}_{T-T_j}(w) f^{\widehat M^i,Y^i}_{T_j-T_i}(u,v)\d w \d v \d u\\
	=&\int_0^1\int_{0}^{\infty}\int_0^1  \int_0^{\infty} \int_0^{\infty} \mbb P(\ell^1 < M_{T_i}, \ell^1 < uX_{T_i}, \ell^2 < wvX_{T_i}|\cal F^D_{T_i})\\
	&\qquad \qquad\qquad f_{L^1,L^2}(\ell^1,\ell^2)f^{\widehat M^j}_{T-T_j}(w) f^{\widehat M^i,Y^i}_{T_j-T_i}(u,v)\d \ell^2 \d \ell^1\d w \d v \d u\\
	=&\int_0^1\int_{0}^{\infty}\int_0^1  \int_0^{uX_{T_i}} \int_0^{wvX_{T_i}} \mbb P(\ell^1 < M_{T_i}|\cal F^D_{T_i}) f_{L^1,L^2}(\ell^1,\ell^2)f^{\widehat M^j}_{T-T_j}(w) f^{\widehat M^i,Y^i}_{T_j-T_i}(u,v)\d \ell^2 \d \ell^1\d w \d v \d u,
	\end{align*}	
	where we have used in the second to last equation that $(L^1,L^2)$ is independent of $\cal F_T$. Since $Y^i_u \overset{d}{=}X_u$, for $i=1,\ldots,J-1$, $f_t^{\widehat M^i}$ and $f_t^{\widehat M^i, Y^i}$ are given in \eqref{eq:density_min} and \eqref{eq:joint_density_min_x}, respectively. Before we calculate the probability in the above integrand we make the following notation
	\begin{align*}
	P_i(\ell) = \mbb P(\ell < M_{T_i}|X_{T_0},\ldots,X_{T_i}).
	\end{align*}
	It holds
	\begin{align*}
	\{\ell < M_{T_i}\}&=\{\ell < \inf_{s<T_i}X_s\}=\{\ell < \inf_{s<T_{i-1}}X_s\}\cap\{\ell < \inf_{T_{i-1}\leq s<T_i}X_s\}=\{\ell < M_{T_{i-1}}\}\cap\{\ell <M_{[T_{i-1},T_i)}\}.
	\end{align*}
	We have
	\begin{align*}
	P_i(\ell^1) =& \mbb P(\ell^1 < M_{T_i}|X_{T_0},\ldots,X_{T_i})=\mbb E[\1{\ell^1 < M_{T_i}}|X_{T_0},\ldots,X_{T_i}]\\
	=&\mbb E[\1{\ell^1 < M_{T_{i-1}}}\1{\ell^1 < M_{[T_{i-1},T_i)}}|X_{T_0},\ldots,X_{T_i}]\\
	=&\mbb E[\mbb E[\1{\ell^1 < M_{T_{i-1}}}\1{\ell^1 < M_{[T_{i-1},T_i)}}| X_s: s\leq T_{i-1},X_{T_i}]|X_{T_0},\ldots,X_{T_i}],
	\end{align*}
	where the last equation follows from the tower property of the conditional expectation since 
	\begin{align*}
	\sigma(X_{T_0},\ldots,X_{T_i}) \subset \cal F_{T_{i-1}} \vee \sigma(X_{T_i}).
	\end{align*}
	We obtain
	\begin{align*}
	P_i(\ell^1)=&\mbb E[\mbb E[\1{\ell^1 < M_{T_{i-1}}}\1{\ell^1 < M_{[T_{i-1},T_i)}}| X_s: s\leq T_{i-1},X_{T_i}]|X_{T_0},\ldots,X_{T_i}]\\
	=&\mbb E[\1{\ell^1 < M_{T_{i-1}}}\mbb E[\1{\ell^1 < M_{[T_{i-1},T_i)}}| X_s: s\leq T_{i-1},X_{T_i}]|X_{T_0},\ldots,X_{T_i}]\\
	=&\mbb E[\1{\ell^1 < M_{T_{i-1}}}\mbb P(\ell^1 < M_{[T_{i-1},T_i)}| X_s: s\leq T_{i-1},X_{T_i})|X_{T_0},\ldots,X_{T_i}],
	\end{align*}
	where we have used that $M_{T_{i-1}}$ is $\cal F_{T_{i-1}}$-measurable. The probability in the above equation is rewritten as
	\begin{align*}
	\mbb P(\ell^1 < M_{[T_{i-1},T_i)}| X_s: s\leq T_{i-1},X_{T_i})=&\mbb P(\ell^1 < \inf_{T_{i-1}\leq s<T_i}X_s| X_s: s\leq T_{i-1},X_{T_i})\\
	=&\mbb P(\ell^1 < \inf_{T_{i-1}\leq s<T_i}X_{T_{i-1}}Y^{i-1}_{s-T_{i-1}}| X_s: s\leq T_{i-1},X_{T_i})\\
	=&\mbb P(\ell^1/X_{T_{i-1}} < \widehat M^{i-1}_{T_i-T_{i-1}}| X_s: s\leq T_{i-1},X_{T_i})\\
	=&\mbb P(\ell^1/X_{T_{i-1}} < \widehat M^{i-1}_{T_i-T_{i-1}}| X_{T_{i-1}},X_{T_i}),
	\end{align*}
	where last equation holds since $\widehat M^{i-1}_{T_i-T_{i-1}}$ is independent of $\cal F_{T_{i-1}}$. Finally, we obtain the following recursion formula
	\begin{align*}
	P_i(\ell^1)=&\mbb E[\1{\ell^1 < M_{T_{i-1}}}\mbb P(\ell^1 < M_{[T_{i-1},T_i)}| X_s: s\leq T_{i-1},X_{T_i})|X_{T_0},\ldots,X_{T_i}]\\
	=&\mbb E[\1{\ell^1 < M_{T_{i-1}}}\mbb P(\ell^1/X_{T_{i-1}} < \widehat M^{i-1}_{T_i-T_{i-1}}| X_{T_{i-1}},X_{T_i})|X_{T_0},\ldots,X_{T_i}]\\
	=&\mbb E[\1{\ell^1 < M_{T_{i-1}}}|X_{T_0},\ldots,X_{T_{i-1}}]\mbb P(\ell^1/X_{T_{i-1}} < \widehat M^{i-1}_{T_i-T_{i-1}}| X_{T_{i-1}},X_{T_i})\\
	=&\mbb P(\ell^1 < M_{T_{i-1}}|X_{T_0},\ldots,X_{T_{i-1}})\mbb P(\ell^1/X_{T_{i-1}} < \widehat M^{i-1}_{T_i-T_{i-1}}| X_{T_{i-1}},X_{T_i})\\
	=&P_{i-1}(\ell^1)\mbb P(\ell^1/X_{T_{i-1}} < \widehat M^{i-1}_{T_i-T_{i-1}}| X_{T_{i-1}},X_{T_i}).
	\end{align*}
	Lemma \ref{lem:cond_prob_min} yields
	\begin{align*}
	\mbb P(\ell^1/X_{T_{i-1}} < \widehat M^{i-1}_{T_i-T_{i-1}}\,|\, X_{T_{i-1}},X_{T_i})=&1-\exp\left\{\frac{-2}{\sigma^2 (T_i-T_{i-1})}\ln\Big(\frac{\ell^1}{X_{T_{i-1}}}\Big)\ln\Big(\frac{\ell^1}{X_{T_i}}\Big)\right\}
\end{align*}
	for $\ell^1< \min(X_{T_{i-1}},X_{T_i})$ and zero otherwise. Note that
	\begin{align*}
	K_i(\ell) = 
	1-\exp\left\{\frac{-2}{\sigma^2 (T_i-T_{i-1})}\ln\left(\frac{\ell}{X_{T_{i-1}}}\right)\ln\left(\frac{\ell}{X_{T_i}}\right)\right\}, 
	\end{align*}
	for $\ell< \min(X_{T_{i-1}},X_{T_i})$ and $K_i(\ell) =0$ otherwise.
	Then the probability $P_i(\ell^1)$ can be recursively calculated by
	\begin{align*}
	P_i(\ell^1) &= \mbb P(\ell^1 < M_{T_i}|X_{T_0},\ldots,X_{T_i})=
	P_{i-1}(\ell^1)\mbb P(\ell^1/X_{T_{i-1}} < \widehat M^{i-1}_{T_i-T_{i-1}}| X_{T_{i-1}},X_{T_i})\\
	&=P_{i-1}(\ell^1)K_i(\ell^1)=P_{i-2}(\ell^1)K_{i-1}(\ell^1)K_i(\ell^1)=\ldots= \prod_{j=1}^i K_j(\ell^1),
	\end{align*}
	since $P_0(\ell^1)= \mbb P(\ell^1 < M_{T_0}|X_{T_0})=\mbb P(\ell^1 < X_{T_0}|X_{T_0})=\1{\ell^1<X_{T_0}}$.\\\\
	Eventually we obtain for $t\in [t_0,t_1)$ and $t\in [T_i,T_{i+1})$ that 
	\begin{align*}
	\mbb P(\tau >T|\cal F^D_t) &=\int_0^1 \int_{0}^{\infty} \int_0^1  \int_0^{uX_{T_i}} \int_0^{wvX_{T_i}} P_i(\ell^1)f_{L^1,L^2}(\ell^1,\ell^2)f^{\widehat M^j}_{T-T_j}(w)  f^{\widehat M^i,Y^i}_{T_j-T_i}(u,v)\d \ell^2 \d \ell^1\d w \d v \d u.	
	\end{align*}
	\myparagraph{Proof of \eqref{eq:sp_d_b}:} 
	Let $t\in [t_1,T)$ and $t\in [T_i,T_{i+1})$ with $i\geq j$, i.e., $t_1=T_j \leq T_{i}$. Then
	\begin{align*}
	\{\tau >T\}=&\{L^1 < M_{t_1}\}\cap \{ L^2 < M_{[t_1,T)} \}= \{L^1 < M_{T_j}\}\cap \{ L^2 < M_{[T_j,T)} \}\\
	=&\{L^1 < M_{T_j}\}\cap \{ L^2 < M_{[T_j,T_i)} \}\cap \{ L^2 < M_{[T_i,T)} \}\\
	=&\{L^1 < M_{T_j}\}\cap \{ L^2 < M_{[T_j,T_i)} \}\cap \{ L^2 < \inf_{T_i\leq s < T} X_s \}\\
	=&\{L^1 < M_{T_j}\}\cap \{ L^2 < M_{[T_j,T_i)} \}\cap \{ L^2 < \inf_{T_i\leq s < T} X_{T_i} Y^i_{s-T_i}\}\\
	=&\{L^1 < M_{T_j}\}\cap \{ L^2 < M_{[T_j,T_i)} \}\cap \{ L^2 < X_{T_i} \widehat M^i_{T-T_i}\}
	\end{align*}
	and we obtain
	\begin{align*}
	\mbb P(\tau >T|\cal F^D_t)=&\mbb P(\tau >T|\cal F^D_{T_i}) =\mbb P(\tau >T|X_{T_0},\ldots,X_{T_i})\\
	=&\mbb P(L^1 < M_{T_j}, L^2 <M_{[T_j,T_i)}, L^2 < \widehat M^i_{T-T_i}X_{T_i}|X_{T_0},\ldots,X_{T_i})\\
	=&\int\limits_0^\infty\int\limits_0^\infty\mbb P(\ell^1 < M_{T_j}, \ell^2 <M_{[T_j,T_i)}, \ell^2 < \widehat M^i_{T-T_i}X_{T_i}|X_{T_0},\ldots,X_{T_i})f_{L^1,L^2}(\ell^1,\ell^2)\d \ell^2 \d \ell^1\\
	=&\int\limits_0^\infty\int\limits_0^\infty \mbb P(\ell^1 <M_{T_j},\ell^2 <M_{[T_j,T_i)}|X_{T_0},\ldots,X_{T_i}) \mbb P(\ell^2/X_{T_i} < \widehat M^i_{T-T_i}) f_{L^1,L^2}(\ell^1,\ell^2)\d \ell^2 \d \ell^1\\
	=&\int\limits_0^\infty\int\limits_0^\infty \mbb P(\ell^1 <M_{T_j},\ell^2 <M_{[T_j,T_i)}|X_{T_0},\ldots,X_{T_i})\Psi\left(T-T_i,\frac{\ell^2}{X_{T_i}}\right) f_{L^1,L^2}(\ell^1,\ell^2)\d \ell^2 \d \ell^1\\
	=&\int\limits_0^\infty\int\limits_0^\infty  \prod_{k=1}^j K_k(\ell^1) \prod_{k=j+1}^i K_k(\ell^2)\Psi\left(T-T_i,\frac{\ell^2}{X_{T_i}}\right)f_{L^1,L^2}(\ell^1,\ell^2)\d \ell^2 \d \ell^1\\
	=&\int\limits_0^\infty\int\limits_0^\infty  P_j(\ell^1) \prod_{k=j+1}^i K_k(\ell^2)\Psi\left(T-T_i,\frac{\ell^2}{X_{T_i}}\right)f_{L^1,L^2}(\ell^1,\ell^2)\d \ell^2 \d \ell^1,
	\end{align*}
	where $\Psi(t,\cdot)$ is the complementary distribution function of $M_t$ given in Lemma \ref{compl_distr}. 
	\\
	\myparagraph{Proof of \eqref{eq:sp_d_c}:} 
	Let $t\in [t_0,t_1)$ and $t\in [T_i,T_{i+1})$. Then
	\begin{align*}
	\mbb P(\tau > t |\cal F^D_t)=& \mbb P(L^1 < M_{t}|\cal F^D_t) = \mbb P(L^1 < M_{t}|\cal F^D_{T_i})=\mbb P(L^1 < M_{t}|X_{T_0},\ldots,X_{T_i})\\
	=&\mbb P(L^1 < M_{T_i}, L^1 < M_{[T_i,t)}|X_{T_0},\ldots,X_{T_i})=\mbb P(L^1 < M_{T_i}, L^1 < \widehat M^i_{t-T_i}X_{T_i}|X_{T_0},\ldots,X_{T_i})\\
	=&\int_0^{\infty} \mbb P(\ell^1 < M_{T_i}, \ell^1 < \widehat M^i_{t-T_i}X_{T_i}|X_{T_0},\ldots,X_{T_i}) f_{L^1}(\ell^1)\d \ell^1\\
	=&\int_0^{\infty} \mbb P(\ell^1 < M_{T_i}|X_{T_0},\ldots,X_{T_i})\mbb P(\ell^1/X_{T_i} < \widehat M^i_{t-T_i}) f_{L^1}(\ell^1)\d \ell^1\\
	=&\int_0^{\infty}  P_i(\ell^1) \Psi\left(t-T_i,\frac{\ell^1}{X_{T_i}}\right)f_{L^1}(\ell^1)\d \ell^1. 
	\end{align*}
	The second to last equation follows by the independence from $\widehat M^i_{t-T_i}$ of $T_i$ and the last equation holds since $\widehat M_{t-T_i}^i \overset{d}{=} M_{t-T_i}$.\\ 
	For $t\in [t_1,T)$ and $t\in [T_i,T_{i+1})$ with $i\geq j$, i.e., $t_1=T_j \leq T_{i}$, it holds 
	\begin{align*}
	\{\tau >T\}&=\{L^1 < M_{t_1}\}\cap \{ L^2 < M_{[t_1,t)} \}= \{L^1 < M_{T_j}\}\cap \{ L^2 < M_{[T_j,t)} \}\\
	&=\{L^1 < M_{T_j}\}\cap \{ L^2 < M_{[T_j,T_i)} \}\cap \{ L^2 < M_{[T_i,t)} \}\\
	&=\{L^1 < M_{T_j}\}\cap \{ L^2 < M_{[T_j,T_i)} \}\cap \{ L^2 < \inf_{T_i\leq s < t} X_s \}\\
	&=\{L^1 < M_{T_j}\}\cap \{ L^2 < M_{[T_j,T_i)} \}\cap \{ L^2 < \inf_{T_i\leq s < t} X_{T_i} Y^i_{s-T_i}\}\\
	&=\{L^1 < M_{T_j}\}\cap \{ L^2 < M_{[T_j,T_i)} \}\cap \{ L^2 < X_{T_i} \widehat M^i_{t-T_i}\}.
	\end{align*}
	Using the same arguments as above yields		
	\begin{align*}
	\mbb P(\tau >t|\cal F^D_t) =&\mbb P(\tau >t|\cal F^D_{T_i}) =\mbb P(\tau >t|X_{T_0},\ldots,X_{T_i})\\
	=&\mbb P(L^1 < M_{T_j}, L^2 <M_{[T_j,T_i)}, L^2 < \widehat M^i_{t-T_i}X_{T_i}|X_{T_0},\ldots,X_{T_i})\\
	=&\int\limits_0^\infty\int\limits_0^\infty\mbb P(\ell^1 < M_{T_j}, \ell^2 <M_{[T_j,T_i)}, \ell^2 < \widehat M^i_{t-T_i}X_{T_i}|X_{T_0},\ldots,X_{T_i}) f_{L^1,L^2}(\ell^1,\ell^2)\d \ell^2 \d \ell^1\\
	=&\int\limits_0^\infty\int\limits_0^\infty \mbb P(\ell^1 <M_{T_j},\ell^2 <M_{[T_j,T_i)}|X_{T_0},\ldots,X_{T_i})\mbb P(\ell^2/X_{T_i} < \widehat M^i_{t-T_i}) f_{L^1,L^2}(\ell^1,\ell^2)\d \ell^2 \d \ell^1\\
	=&\int\limits_0^\infty\int\limits_0^\infty  P_j(\ell^1) \prod_{k=j+1}^i K_k(\ell^2)\Psi\left(t-T_i,\frac{\ell^2}{X_{T_i}}\right)f_{L^1,L^2}(\ell^1,\ell^2)\d \ell^2 \d \ell^1.
	\end{align*}

\bibliographystyle{abbrv}

\end{document}